\documentclass{my_article_class}
\usepackage{settings}

\title{\bf\fontsize{20.00}{25.00}\selectfont 
	Pauli Encodings \& Unclonable Encryption
}

\author{Pierre Botteron
	\thanks{pierre.botteron@ens-lyon.fr, $\qquad^\dagger$sebastien.designolle@inria.fr, $\qquad^\ddagger$omar.fawzi@ens-lyon.fr}}
\author{Sébastien Designolle$^\dagger\!\!$}
\author{Omar Fawzi$^\ddagger$}

\affil{INRIA, ENS de Lyon, Lyon, France}

\date{}

\begin{document}

\maketitle

\vspace{-0.5cm}
\begin{abstract}
	The unclonable bit question asks whether quantum encryption can prevent an adversary from producing two systems that both reveal the plaintext once the key is disclosed.
	We introduce and study \emph{Pauli Encodings}, a simple class of one-bit encryption schemes whose ciphertexts are normalized eigenspace projectors of Pauli strings.
	For every Pauli Encoding with $K$ Pauli strings, we prove a universal lower bound $1/2+1/(2\sqrt{K})$ on the optimal monogamy-of-entanglement winning probability, together with sharper bounds for several structured families.
	We then establish two limitations of natural approaches to unclonable security.
	First, if the Pauli strings are restricted to strings of $X$ and $Z$ of length $n$, the encoding is not secure.
	Second, we identify a universal 3/4 obstruction showing that arguments based only on pairwise guessing marginals cannot establish unclonable-indistinguishable security.
	When the Pauli strings all pairwise anticommute, the protocol becomes the one studied in~\href{https://doi.org/10.22331/q-2026-07-08-2157}{[Quantum 10, 2157 (2026)]}.
	We exploit the symmetry of this protocol to solve the third level of the natural semidefinite programming relaxation obtaining an asymptotic upper bound of approximately 0.5556 on the winning probability.
	Finally, we prove strong unclonable-indistinguishable security against bounded-local-dimension adversaries and strong indistinguishability security for several efficient Pauli families.
	First-level NPA computations provide additional numerical evidence towards the strong unclonable-indistinguishable security.
\end{abstract}
\vspace{0.5cm}

\small
\setcounter{tocdepth}{2}
\tableofcontents{}
\normalsize

\section{Introduction}

Classical information can be copied, even when it is encrypted.
By contrast, the no-cloning theorem \cite{Dieks-82,Wootters-Zurek-82} states that an unknown quantum state cannot be perfectly cloned.
This principle can be leveraged to obtain unclonable security guarantees, opening the way to unclonable cryptography.

\paragraph{The Unclonable Bit Question.}
This line of research was initiated by \citeauthor{Gottesman-03} in \cite{Gottesman-03} and further developed by \citeauthor{Broadbent-Lord-20} in \cite{Broadbent-Lord-20}.
A central question is whether a classical bit\footnote{
	It is sufficient to study one-bit protocols because they can be transformed into $k$-bit protocols with the same security \cite{Hiroka-Kitagawa-Nishimaki-Yamakawa-23}.
	However, unlike in other cryptographic settings, repeatedly encoding a one-bit message does not trivially yield security for a string of bits.
} $x\in\set{0,1}$ can be encrypted under a classical key~$k$ as a quantum state $\rho_{x|k}$ while satisfying
correctness (the message can be recovered from the key and the ciphertext),
efficiency (the algorithms run in polynomial time),
and an unclonable property called \emph{unclonable-indistinguishability} (\Cref{def:Security}).
Informally, an adversary first applies a quantum channel~$\Phi$ to split the ciphertext $\rho_{x|k}$ between two noncommunicating parties, Bob and Charlie.
After the key is revealed, security requires that Bob and Charlie cannot both recover Alice's message~$x$ with probability substantially greater than~$1/2$.
A protocol satisfying these requirements is called an \emph{unclonable bit}.

\paragraph{Motivation.}
The unclonable bit question is important because of its applications to several cryptographic primitives,
including private-key quantum money \cite{Broadbent-Lord-20},
quantum functional encryption \cite{Mehta-Muller-24},
quantum copy-protection \cite{Ananth-Kaleoglu-21},
quantum position verification \cite{George-Allerstorfer-Lunel-Chitambar-25},
and uncloneable decryption \cite{Georgiou-Zhandry-20,Sattath-Wyborski-22,Kundu-Tan-22}.

\paragraph{Prior Work.}
Many approaches have been proposed since the initial formulation of the problem.
The question has been studied in the quantum random oracle model (QROM) \cite{Broadbent-Lord-20,Ananth-Kaleoglu-Li-Liu-Zhandry-22,Ananth-Kaleoglu-Liu-23},
in an interactive setting \cite{Broadbent-Culf-23},
in device-independent variants \cite{Kundu-Tan-22},
under obfuscation assumptions \cite{Ananth-Behera-24,Chevalier-Hermouet-Vu-24},
with quantum keys \cite{Ananth-Kaleoglu-Yuen-24},
through the relaxation known as untelegraphable-indistinguishable security \cite{Broadbent-Culf-Rochette-25},
and in the Haar random oracle model \cite{Bartusek-Goldin-26}.
Several limitations have also been established \cite{Majenz-Schaffner-Tahmasbi-21,Ananth-Kaleoglu-Li-Liu-Zhandry-22,Coladangelo-Liu-Xie-25}.

More recently, a candidate scheme based on anticommuting Pauli strings was proposed in \cite{Botteron-Broadbent-Culf-Nechita-Pellegrini-Rochette-24}.
Its conjectured security was proved for $2\leq K\leq7$ and verified numerically up to $K=17$.
A different construction based on Haar-random unitaries was then shown to be efficient and weakly secure in \cite{Bhattacharyya-Culf-26}, and subsequently strongly secure but inefficient in \cite{Bhattacharyya-Broadbent-Culf-26}.
Finally, concurrently with and independently of this manuscript, two works established strong security and efficiency for schemes based on Pauli strings \cite{Ananth-Sahai-26,Ragavan-26}.
More detailed comparisons with these works are given in \Cref{subsec:Comparison-with-Recent-Works}.

\subsection{Contributions}

In this work, we generalize the anticommuting protocol of~\cite{Botteron-Broadbent-Culf-Nechita-Pellegrini-Rochette-24} to a broad class of Pauli protocols, which we call \emph{Pauli Encodings}.
For a subset
$\PauliSet^{(n)}\subseteq\set{\Identity_2,X,Y,Z}^{\otimes n}\backslash\set{\Identity_{2^n}}$
(or, more generally, a multiset of Pauli strings), we consider encryption schemes of the form
$$
	\rho_{x|P}^{(n)}
	\,\coloneqq\,
	\frac{\Identity_{2^n}+(-1)^xP}{2^n}\,,
$$
where $n$ is the number of qubits,
the message $x$ belongs to $\set{0,1}$,
and the key $P$ belongs to $\PauliSet^{(n)}$.
We write $K\coloneqq\abs{\PauliSet^{(n)}}$ for the number of keys, counting multiplicities.
These protocols are perfectly correct.
Given the key~$P$, ciphertext preparation and decryption are efficient, while the efficiency of key generation depends on the choice of $\PauliSet^{(n)}$.
We exclude $P=\Identity_{2^n}$ because $\rho_{x|P}^{(n)}$ would not be a quantum state for both values of~$x$.

Following a standard reduction, we upper-bound the winning probability of the cloning game by that of the monogamy-of-entanglement (MoE) game introduced in \cite{Tomamichel-Fehr-Kaniewski-Wehner-13}.
The MoE winning probability is the main quantity studied in this manuscript.

\subsubsection{Lower Bounds}

We first improve the elementary lower bound $1/2+1/(2K)$, obtained by targeting a single key, to a universal bound valid for every Pauli Encoding.

\begin{result}[Informal, \Cref{cor:Universal-Lower-Bound}]
	\label{result:Universal-Lower-Bound}
	For every Pauli Encoding, the MoE winning probability satisfies
	$$
		\Prob^*\of[\big]{\text{\normalfont win MoE}}
		\,\geq\,
		\frac12+\frac{1}{2\sqrt K}\,.
	$$
	This exactly matches the conjectured optimal value for the anticommuting family $\Kanticomm$ \cite{Botteron-Broadbent-Culf-Nechita-Pellegrini-Rochette-24}.
	Thus, if the conjecture holds, the anticommuting protocol is optimal among Pauli Encodings with a fixed number~$K$ of keys.
\end{result}

We then refine this lower bound for several families.\footnote{
	Although the identity cannot be used as an encryption key, it may be adjoined in the MoE formulation.
	Indeed, \Cref{prop:Pwin-of-P_n-union-Q_n} gives an exact affine relation between the MoE winning probabilities of $\PauliSet^{(n)}$ and $\PauliSet^{(n)}\sqcup\set{\Identity_{2^n}}$.
	In particular, the two values are asymptotically equivalent when $\abs{\PauliSet^{(n)}}\to\infty$.
}

\begin{result}[Informal, \Cref{subsec:Improved-Lower-Bounds}]
	For the following Pauli families, we obtain lower bounds stronger than \Cref{result:Universal-Lower-Bound}:
	\begin{itemize}
		\item For
		$
		\Kall^{(n)}\sqcup\set{\Identity_{2^n}}
		=
		\set[\big]{\Identity_2,X,Y,Z}^{\otimes n},
		$
		the lower bound is
		$
		\frac12+\frac12\of[\big]{\frac{1+\sqrt3}{4}}^n.
		$
		\item For
		$
		\Kreal^{(n)}\sqcup\set{\Identity_{2^n}}
		=
		\set[\big]{\Identity_2,X,Z}^{\otimes n},
		$
		the lower bound is
		$
		\frac12+\frac12\of[\big]{\frac{1+\sqrt2}{3}}^n.
		$
		\item For the augmented HZH multiset
		$$
			\set[\Big]{
				H^\theta Z^aH^\theta
				\,:\,
				\theta,a\in\set{0,1}^n
			},
		$$
		which contains $2^n$ identity labels, the lower bound is
		$
		\frac12+\frac12\of[\big]{\frac12+\frac{1}{2\sqrt2}}^n.
		$
	\end{itemize}
	None of our numerical seesaw searches found a larger lower bound.
\end{result}

\subsubsection{No-Go Results}

We establish two no-go results for natural approaches.
The first concerns a pairwise relaxation of the MoE game.
For a tripartite state $\sigma_{\Alice\Bob\Charlie}$, the \emph{pairwise MoE game} averages the two marginal guessing probabilities
$\Prob\of{x_\Alice=x_\Bob\,|\,\sigma_{\Alice\Bob\Charlie}}$
and
$\Prob\of{x_\Alice=x_\Charlie\,|\,\sigma_{\Alice\Bob\Charlie}}$.
This relaxation gives a natural upper bound on the usual MoE winning probability, but the following result shows that it cannot prove strong security.
We found this no-go result because some previous versions of papers in the community relied on this natural upper bound.
We refer to this obstruction as the ``curse of $3/4$.''

\begin{result}[Informal, \Cref{prop:Curse-of-the-3_4}]
	For every correct binary encryption scheme, the optimal winning probability of the pairwise MoE game is at least~$3/4$.
\end{result}

We also give an alternative proof, in the Pauli-Encoding formalism, of the exact BB84 result previously established in \cite{Coladangelo-Liu-Xie-25}.

\begin{result}[Informal, \Cref{prop:Non_Security_of_the_BB84_Encoding}]
	\label{result:Non_Security_of_the_BB84_Encoding}
	For every $n\in\NN^*$, the BB84 Encoding
	$
	\KBBeightyfour^{(n)}
	\coloneqq
	\set{X,Z}^{\otimes n}
	$
	satisfies
	$$
		\Prob^*\of[\Big]{
			\text{\normalfont win MoE}
			\,\Big|\,
			\KBBeightyfour^{(n)}
		}
		=
		\cos^2\of[\Big]{\frac{\pi}{8}}
		=
		\frac12+\frac{1}{2\sqrt2}\,.
	$$
	Consequently, the BB84 Encoding is not weakly unclonable-indistinguishable secure.
\end{result}

\subsubsection{Security Results}

We next establish several positive security results.
First, for bounded-dimensional adversaries, a Frobenius-norm argument yields the following statement.

\begin{result}[\Cref{lem:upper-bound-in-terms-of-Bobs-dimension}]
	Denote by $d\in\NN^*$ a common upper bound on the dimensions of Bob's and Charlie's Hilbert spaces.
	Then, for any $n\in\NN^*$ and any $\PauliSet^{(n)}\subseteq\set{\Identity_2,X,Y,Z}^{\otimes n}$, we have:
	$$
		\Prob^*\of[\big]{ \text{\normalfont win MoE} }
		\,\leq\,
		\frac12 + \frac{d\,\sqrt{2^n}}{2\sqrt{\abs{\PauliSet^{(n)}}}}\,.
	$$
\end{result}

We then prove strong indistinguishability security, a necessary but insufficient condition for unclonable-indistinguishability.

\begin{result}[Informal, \Cref{prop:Strong_Indistinguishable_Security}]
	The families $\Kall^{(n)}$, $\Kreal^{(n)}$, and $\KHZH^{(n,\Lall(n))}$ satisfy strong indistinguishability security.
	The BB84 Encoding $\KBBeightyfour^{(n)}$ also satisfies strong indistinguishability security, even though it is not unclonable-indistinguishable secure by \Cref{result:Non_Security_of_the_BB84_Encoding}.
\end{result}

For the anticommuting family, we go beyond the numerical evidence of~\cite{Botteron-Broadbent-Culf-Nechita-Pellegrini-Rochette-24}.
We perform a full symmetry reduction of the NPA moment hierarchy through level three and determine the associated asymptotic relaxations.

\begin{result}[Informal, \Cref{subsubsec:anticommuting-asymptotic-level-three}]
	For the anticommuting protocol $\Kanticomm^{(n)}$, the level-three NPA relaxation yields
	$$
		\limsup_{n\to\infty}
		\Prob^*\of[\Big]{
			\text{\normalfont win MoE}
			\,\Big|\,
			\Kanticomm^{(n)}
		}
		\,\leq\,
		\omega_3^{(\infty)},
		\qquad
		\omega_3^{(\infty)}\approx0.555608131\,.
	$$
\end{result}

The first level of the hierarchy was already studied in \cite{Botteron-Broadbent-Culf-Nechita-Pellegrini-Rochette-24}, which obtained the asymptotic upper bound~$5/8$.
We prove that the anticommuting commutation pattern is optimal at this level.

\begin{result}[Informal, \Cref{prop:anticommuting-best-at-level-one}]
	Among all Pauli commutation patterns with a fixed number~$K$ of keys, the anticommuting pattern minimizes the level-one NPA upper bound.
	Consequently, no level-one analysis of a Pauli protocol can yield an asymptotic upper bound below~$5/8$.
\end{result}

We also present numerical upper bounds for several efficient Pauli Encodings at small values of~$n$ in \Cref{tab:level-one-efficient-Pauli-protocols}.
These computations provide further evidence that the corresponding protocols are promising candidates for the unclonable bit.
Our implementation can be found in the accompanying GitHub repository \cite{GitHub}.

Finally, \Cref{app:anticommuting-NPA-reduction} gives the technical details of the exact symmetry reduction of the NPA hierarchy for the anticommuting protocol $\Kanticomm^{(n)}$.

\paragraph{Future Work.}
Two concurrent and independent works \cite{Ananth-Sahai-26,Ragavan-26} prove strong security for efficient Pauli Encodings.
The construction of \cite{Ananth-Sahai-26} uses
$
\set{X,Y}\otimes\set{\Identity_2,X,Y,Z}^{\otimes(n-1)},
$
whereas \cite{Ragavan-26} proves strong security for
$
\Kall^{(n)}
=
\set{\Identity_2,X,Y,Z}^{\otimes n}\backslash\set{\Identity_{2^n}}.
$
Based on the evidence presented here, we conjecture that strong security and efficiency can also be achieved for the following Pauli Encodings (see \Cref{subsec:Examples-of-Pauli-Encodings}):
\begin{itemize}
	\item
	$
	\Kreal^{(n)}
	\coloneqq
	\set[\big]{\Identity_2,X,Z}^{\otimes n}\backslash\set{\Identity_{2^n}},
	$
	\item
	$
	\KHZH^{(n,\Lall(n))}
	\coloneqq
	\set[\Big]{
		H^\theta Z^aH^\theta
		\,:\,
		\theta\in\set{0,1}^n,\,
		a\in\set{0,1}^n\backslash\set{\zero}
	},
	$
	\item
	$
	\KHZH^{(n,\Lsym(\lceil n/2\rceil,n))}
	\coloneqq
	\set[\Big]{
		H^\theta Z^aH^\theta
		\,:\,
		\theta\in\set{0,1}^n,\,
		a\in\set{0,1}^n,\,
		\abs{a}=\lceil n/2\rceil
	}.
	$
\end{itemize}
It would also be interesting to prove the conjectured strong security of the anticommuting family $\Kanticomm^{(n)}$ for all~$n$, despite its inefficiency \cite{Botteron-Broadbent-Culf-Nechita-Pellegrini-Rochette-24}.

\subsection{Comparison with Recent Works}
\label{subsec:Comparison-with-Recent-Works}

\paragraph{Comparison with \cite{Botteron-Broadbent-Culf-Nechita-Pellegrini-Rochette-24}.}
The work \cite{Botteron-Broadbent-Culf-Nechita-Pellegrini-Rochette-24} studies the Pauli Encoding defined by pairwise anticommuting Pauli strings.
The number of qubits grows linearly with~$K$ and is therefore exponential in the key length $\lambda=\lceil\log_2K\rceil$.
The authors prove the optimal value $1/2+1/(2\sqrt K)$ for $2\leq K\leq7$, verify it numerically up to $K=17$, and provide strong evidence that it holds for every $K\in\NN^*$.
We further study this protocol by numerically verifying the conjectured value up to $K=40$ at level three and by analyzing the asymptotic NPA relaxations; see \Cref{subsec:Security_Analysis_of_the_Anticommuting_Protocol}.

\paragraph{Comparison with \cite{Bhattacharyya-Culf-26,Bhattacharyya-Broadbent-Culf-26}.}
The Haar-unitary construction of \cite{Bhattacharyya-Culf-26} achieves weak security with inverse-polynomial advantage and can be implemented efficiently using unitary designs.
The follow-up work \cite{Bhattacharyya-Broadbent-Culf-26} achieves exponentially small advantage in terms of number of qubits, but its construction is not efficient.
These schemes are not Pauli Encodings and are not studied in this manuscript.

\paragraph{Concurrent and Independent Works \cite{Ananth-Sahai-26,Ragavan-26}.}
In July 2026, two concurrent and independent works established efficient, strongly secure Pauli Encodings.
Ananth and Sahai \cite{Ananth-Sahai-26} use the family
$
\set{X,Y}\otimes\set{\Identity_2,X,Y,Z}^{\otimes(n-1)}
$
and obtain winning probability at most
$
1/2+1/\sqrt{2^{n+1}}.
$
Ragavan \cite{Ragavan-26} uses
$
\Kall^{(n)}
=
\set{\Identity_2,X,Y,Z}^{\otimes n}\backslash\set{\Identity_{2^n}}
$
and proves the upper bound
$
1/2+\frac12\sqrt{2^n/(4^n-1)}.
$
The latter family is also analyzed here, concurrently and independently, although our results establish only partial security guarantees for it.

\subsection{Acknowledgments}

The authors are grateful to
Gilles Brassard,
Anne Broadbent,
Paul Hermouet,
Ion Nechita, and
Denis Rochette
for insightful discussions on unclonable cryptography.
They also thank
Peter Brown and Lewis Wooltorton for helpful comments on the ncpol2sdpa Python package~\cite{Wittek-15},
Nando Leijenhorst for assistance with his SDP-rounding package~\cite{CLL24},
and Vincent Russo for the toqito Python package~\cite{toqito}.
This project received funding from the European Union's Horizon Europe research and innovation program through the Quantum Secure Networks Partnership (QSNP, grant agreement No.~101114043).

\paragraph{Use of AI.}
The idea of generalizing the anticommuting scheme of \cite{Botteron-Broadbent-Culf-Nechita-Pellegrini-Rochette-24} to general Pauli Encodings originated with the authors from May 2026.
The research phase was conducted by the authors and, in some sections, supported by extended discussions with GPT-5.5 Free and Plus and Gemini 3.5 Free,
especially for the results presented in \Cref{subsec:Universal_Lower_Bound,subsec:Strong_Security_Against_Bounded_Adversaries,subsec:Security_Analysis_of_the_Anticommuting_Protocol}.
This manuscript was written by the authors themselves, 
except for the following technical sections: \Cref{subsec:Security_Analysis_of_the_Anticommuting_Protocol,subsec:Security_Analysis_of_some_Efficient_Protocols} and \Cref{app:anticommuting-NPA-reduction} were drafted by ChatGPT through extended discussions and refined by the authors.
AI tools were also used to polish the manuscript.
The authors take full accountability for all results and proofs presented in this work.

\section{Preliminaries}
\label{sec:Preliminaries}

\subsection{Cloning and Monogamy-of-Entanglement Games}

\subsubsection{Definition of the Cloning Games}

In the cloning game, Alice samples a classical message $x\in \Xset$ and a classical key $k\in \Kset$ uniformly at random.
She encrypts~$x$ into a quantum state $\rho_{x|k}\in\Density(\Hilbert_\Alice)$.
She then sends the state to a pirate $\Pirate$ who applies a quantum channel
$$
	\Phi:\Bounded(\Hilbert_\Alice)\to\Bounded(\Hilbert_\Bob\otimes \Hilbert_\Charlie)
$$ 
without knowing the key $k$.
Then Bob and Charlie obtain the key $k$ and perform measurements using their respective POVMs $\set{B_{x_\Bob|k}}_{x_\Bob}$ and $\set{C_{x_\Charlie|k}}_{x_\Charlie}$.
We say that the adversary team $(\Pirate,\Bob,\Charlie)$ wins the cloning game with the strategy $(\Phi, \set{B_{x_\Bob|k}}_{x_\Bob,k}, \set{C_{x_\Charlie|k}}_{x_\Charlie,k})$ if both Bob and Charlie correctly guess the message, that is:
$$
	x_\Bob = x_\Charlie = x\,.
$$
The winning probability is expressed as follows:
\begin{align*}
	\Prob\of[\Big]{\text{$(\Phi, B_{x_\Bob|k}, C_{x_\Charlie|k})$ wins cloning} \,\Big| \set{\rho_{x|k}}}
	&
	\,\coloneqq\,
	\frac{1}{|\Xset|\,|\Kset|} \sum_{x,x_\Bob,x_\Charlie\in \Xset} \sum_{k\in \Kset} \Tr\ofAlt[\Big]{ \Phi\of{\rho_{x|k}} \of[\Big]{ B_{x_\Bob|k}\otimes C_{x_\Charlie|k} } } \Indic_{x_\Bob = x_\Charlie = x}
	\\
	&\,=\,
	\frac{1}{|\Xset|\,|\Kset|} \sum_{x\in \Xset} \sum_{k\in \Kset} \Tr\ofAlt[\Big]{ \Phi\of{\rho_{x|k}} \of[\Big]{ B_{x|k}\otimes C_{x|k} } }\,,
\end{align*}
and the optimal winning probability is the supremum over all possible adversarial attacks:
$$
	\Prob^*\of[\Big]{\text{win cloning} \,\Big| \set{\rho_{x|k}}}
	\,\coloneqq\,
	\sup_{(\Phi, B, C)}
	\Prob\of[\Big]{\text{$(\Phi, B, C)$ wins cloning} \,\Big| \set{\rho_{x|k}}}\,.
$$

\begin{remark}
	The choice of a cloning game is specified by the family of states $\set{\rho_{x|k}}_{x,k}$. If one changes this family, then we have a different cloning game with different optimal winning probability.
\end{remark}

\subsubsection{Upper Bound using the Choi Matrix}

Using the Choi matrix $C_\Phi$ of the quantum channel $\Phi$, we can rephrase the winning probability:
\begin{align*}
	\Prob^*\of[\Big]{\text{win cloning} \,\Big| \set{\rho_{x|k}}}
	&\,=\,
	\sup_{\substack{C_{\Phi}\succcurlyeq\zero\\\Tr_{(\Bob,\Charlie)}\ofAlt{C_\Phi}=\Identity_\Alice}}
	\sup_{(B, C)}
	\frac{1}{|\Xset|\,|\Kset|} \sum_{x\in \Xset} \sum_{k\in \Kset} \Tr\ofAlt[\Big]{ C_\Phi \of[\big]{\rho_{x|k}^\top \otimes B_{x|k}\otimes C_{x|k} } }\,,
	\\
	&\,\leq\,
	\sup_{\substack{C_{\Phi}\succcurlyeq\zero\\\Tr\ofAlt{C_\Phi}=d_\Alice}}
	\sup_{(B, C)}
	\frac{1}{|\Xset|\,|\Kset|} \sum_{x\in \Xset} \sum_{k\in \Kset} \Tr\ofAlt[\Big]{ C_\Phi \of[\big]{\rho_{x|k}^\top \otimes B_{x|k}\otimes C_{x|k} } }\,,
\end{align*}
where we used the relaxation from the partial trace $\Tr_{(\Bob,\Charlie)}\ofAlt{C_\Phi}=\Identity_\Alice$ to the standard trace $\Tr\ofAlt{C_\Phi}=d_\Alice$.

\subsubsection{Link with Monogamy-of-Entanglement Games}

We renormalize the Choi matrix $C_\Phi$ so that we can view it as a quantum state $\sigma_{\Alice\Bob\Charlie}$ and we obtain the following upper bound:
\begin{align}
	\Prob^*\of[\Big]{\text{win cloning} \,\Big| \set{\rho_{x|k}}}
	&\,\leq\,
	\sup_{\substack{\sigma\succcurlyeq\zero\\\Tr\ofAlt{\sigma}=1}}
	\sup_{(B, C)}
	\frac{1}{|\Xset|\,|\Kset|} \sum_{x\in \Xset} \sum_{k\in \Kset} \Tr\ofAlt[\Big]{ d_\Alice\sigma_{\Alice\Bob\Charlie}\, \of[\big]{\rho_{x|k}^\top \otimes B_{x|k}\otimes C_{x|k} } }\,,
	\\
	\label{eq:winning-prob-at-the-MOE-game}
	&\,=\,
	\sup_{\substack{\sigma\succcurlyeq\zero\\\Tr\ofAlt{\sigma}=1}}
	\sup_{(B, C)}
	\frac{1}{|\Kset|} \sum_{x\in \Xset} \sum_{k\in \Kset} \Tr\ofAlt[\Big]{ \sigma_{\Alice\Bob\Charlie}\, \of[\big]{ A_{x|k} \otimes B_{x|k}\otimes C_{x|k} } }\,,
\end{align}
where we used the following notation:
\begin{equation}
	\label{eq:def_of_Axk_in_terms_of_rho}
	A_{x|k}
	\,\coloneqq\,
	\frac{d_\Alice}{|\Xset|}\, \rho_{x|k}^\top\,.
\end{equation}
Now, whenever $A_{x|k}$ is a valid POVM, \ie when the normalization condition $\sum_{x\in \Xset}\frac{d_\Alice}{|\Xset|}\, \rho_{x|k}^\top = \Identity_\Alice$ holds for all $k\in \Kset$, the expression in \cref{eq:winning-prob-at-the-MOE-game} is exactly the optimal winning probability for the monogamy-of-entanglement game (MoE game) associated with the family $\set{A_{x|k}}_{x,k}$ \cite{Tomamichel-Fehr-Kaniewski-Wehner-13}.
Indeed, recall that in the MoE game, the adversaries choose a quantum tripartite state $\sigma_{\Alice\Bob\Charlie}$, then Alice obtains the classical outcome $x$ using her POVM $\set{A_{x|k}}_{x}$, and then Bob and Charlie also perform their POVMs $\set{B_{x_\Bob|k}}_{x_\Bob}$ and $\set{C_{x_\Charlie|k}}_{x_\Charlie}$. The players Bob and Charlie output $x_\Bob$ and $x_\Charlie$, and they win if and only if $x=x_\Bob=x_\Charlie$. 
Hence, the optimal winning probability is exactly the expression in \cref{eq:winning-prob-at-the-MOE-game} and it yields the following relation:
\begin{equation}
	\label{eq:Cloning_upper_bounded_by_MoE}
	\Prob^*\of[\Big]{\text{win cloning} \,\Big| \set{\rho_{x|k}}}
	\,\leq\,
	\Prob^*\of[\Big]{\text{win MoE} \,\Big| \set{A_{x|k}}}\,.
\end{equation}

\subsection{Unclonable Encryption}

We call \emph{encryption scheme} the family of states $\set{\rho_{x|k}}_{x,k}\subseteq \Density(\Hilbert_\Alice)$ used by Alice in the cloning game.

\begin{definition}[Correctness]
	\label{def:Correctness}
	An encryption scheme $\set{\rho_{x|k}}_{x,k}$ is said to be \emph{correct} if there exists a decryption algorithm $\Dec:\Density(\Hilbert_\Alice)\times \Kset\to \Xset$ that enables retrieval of the original message $x\in \Xset$ with probability one:
	$$
		\forall x\in \Xset,
		\forall k\in \Kset,\quad
		\Prob\of[\Big]{ \Dec\of{ \rho_{x|k}, k } = x }
		\,=\,
		1\,.
	$$
	Equivalently, the scheme $\set{\rho_{x|k}}_{x,k}$ is correct if and only if for every $k\in\Kset$, there exists a POVM $\set{D_{x|k}}_{x\in\Xset}\subseteq\Bounded(\Hilbert_\Alice)$ such that:
	$$
		\forall x\in \Xset,
		\forall k\in \Kset,\quad
		\Tr\ofAlt[\Big]{ D_{x|k}\,\,\rho_{x|k} }
		\,=\,
		1\,.
	$$
\end{definition}

\begin{definition}[Efficiency]
	\label{def:Efficiency}
	A family of correct encryption schemes $\set[\big]{\rho_{x|k}^{(\lambda)}}_{\lambda\in\NN}$ is said to be \emph{efficient} if the key sampling, the encrypted message preparation, and the decryption can be implemented in polynomial time in $\lambda$, where $\lambda$ is called the \emph{security parameter}.
\end{definition}

\noindent
We are interested in the following definition of security, formally introduced by \citeauthor{Broadbent-Lord-20} in \cite{Broadbent-Lord-20}:

\begin{definition}[Unclonable Security]
	\label{def:Security}
	A family of encryption schemes $\set[\big]{\rho_{x|k}^{(\lambda)}}_{\lambda\in\NN}$ is said to be \emph{weakly unclonable secure} if there exists a function $f:\NN\to[0,1]$ with the vanishing condition $f(\lambda)\to0$ as $\lambda\to\infty$  such that:
	$$
		\forall \lambda\in\NN,\qquad
		\Prob^*\of[\Big]{\text{win cloning} \,\Big| \set[\big]{\rho_{x|k}^{(\lambda)}}}
		\,\leq\,
		\frac{1}{|\Xset|} + f(\lambda)\,.
	$$
	If in addition $f$ is negligible, meaning that for every polynomial $p:\NN\to\mathbb R_{>0}$ there exists $\lambda_0\in\NN$ such that $f(\lambda)\leq 1/p(\lambda)$ for all $\lambda\geq\lambda_0$, then we say that the encryption scheme is \emph{strongly unclonable secure}.
	\\
	In the particular case of binary messages $\Xset=\set{0,1}$, the security is called (weak/strong) \emph{unclonable-indistinguishable security}.
\end{definition}


\noindent
Studying the unclonable-indistinguishable security is enough since it implies the more general unclonable security~\cite{Broadbent-Lord-20}. Moreover, it also implies the indistinguishable security~\cite{Broadbent-Lord-20}, which is a standard cryptographic notion.

\begin{definition}[Unclonable Bit]
	The \emph{unclonable bit} is a family of encryption schemes $\set[\big]{\rho_{x|k}^{(\lambda)}}_{\lambda\in\NN} \subseteq \Density\of[\big]{\Hilbert_\Alice^{(\lambda)}}$ that is correct, efficient, and unclonable-indistinguishable secure.
\end{definition}

\begin{question}[\cite{Broadbent-Lord-20}]
	Is the unclonable bit achievable in the plain model?
\end{question}

\section{Candidate Scheme: Pauli Encodings}

Fix a bit $x\in \Xset = \set{0,1}$
and an integer $n\in\NN^*$.
Consider a set $\PauliSet^{(n)} \subseteq \set{\Identity_2,X,Y,Z}^{\otimes n}\backslash\set{\Identity_{2^n}}$\footnote{
    One can more generally assume that $\PauliSet$ is a multiset, \ie it possibly contains several copies of the same element, like the HZH Encoding $\KHZH$ defined in \Cref{ex:HZH_Encoding}.
    In that case, all sums and cardinalities count multiplicities.
} of length-$n$ Pauli strings, which we simply write $\PauliSet$ when the context is clear.
Write $K^{(n)}\coloneqq|\PauliSet^{(n)}|$, or simply $K$ when the context is clear.
Consider the security parameter $\lambda \coloneqq \lceil\log_2 K^{(n)}\rceil$ and keys of the form $k=P\in \PauliSet$.
Recall the expression of the Pauli matrices:
$$
    \Identity_2
    \,\coloneqq\,
    \begin{bmatrix}
        1 & 0 \\
        0 & 1
    \end{bmatrix}
    \,,
    \qquad
    X
    \,\coloneqq\,
    \begin{bmatrix}
        0 & 1 \\
        1 & 0
    \end{bmatrix}
    \,,
    \qquad
    Y
    \,\coloneqq\,
    \begin{bmatrix}
        0 & -i \\
        i & 0
    \end{bmatrix}
    \,,
    \qquad
    Z
    \,\coloneqq\,
    \begin{bmatrix}
        1 & 0 \\
        0 & -1
    \end{bmatrix}
    \,.
$$

		\subsection[Definition]{Definition of the Pauli Encodings}
		\label{section:Definition-of-the-Scheme-with-all-Pauli-strings}


\paragraph{Motivation.}
We want to define a natural family of POVMs $\set[\big]{A_{x|k}^{(n)}}_x\subseteq\Bounded\of[\big]{\Complex^{2^n}}$ for the encryption scheme. 
For simplicity, it is natural to assume they are projective measurements (PVMs).
Now, as they are dichotomic measurements ($x\in\set{0,1}$), they can always be written under the following form:
$$
	A_{x|k}^{(n)}
	\,=\,
	\frac{\Identity_{2^n} + (-1)^x U_k}{2}\,,
$$
for some Hermitian unitaries $U_k\in\Unitary\of[\big]{\Complex^{2^n}}$.
In this work, we find it natural to choose the more specific family of $U_k$ given by Pauli strings:

\begin{definition}[Pauli Encoding, POVM Version]
    \label{def:Pauli_Encoding_POVM_version}
	Consider the following family of POVMs:
	\begin{equation}
		A_{x|P}^{(n)}
		\,\coloneqq\,
		\frac{\Identity_{2^n} + (-1)^x P}{2}
		\in\Bounded\of[\big]{\Complex^{2^n}}
		\,,
	\end{equation}
	where $x\in\set{0,1}$ and $P\in\PauliSet^{(n)}$.
\end{definition}

\noindent
For convenience, when the context is clear, we may simply write $A_{x|P}$.
Note that this family satisfies the positivity condition $A_{x|P}\succcurlyeq\zero$ and the normalization condition $\sum_xA_{x|P} = \Identity_{2^n}$ for all~$P$, so it is a valid POVM.
Then, using \cref{eq:def_of_Axk_in_terms_of_rho}, and up to the key-dependent outcome relabelling explained below, the associated scheme for the cloning game is defined as follows:

\begin{definition}[Pauli Encoding, State Version]
    \label{def:Pauli_Encoding_state_version}
	Consider the following family of quantum states:
	\begin{equation}
        \label{eq:Pauli_Unclonable_Bit}
		\rho_{x|P}^{(n)}
		\,\coloneqq\,
		\frac{\Identity_{2^n} + (-1)^x P}{2^n}
		\in\Density\of[\big]{\Complex^{2^n}}
		\,,
	\end{equation}
	where $x\in\set{0,1}$ and $P\in\PauliSet^{(n)}$.
\end{definition}

\noindent
Again, we may simply write $\rho_{x|P}$ when the context is clear.
We have $\rho_{x|P}\succcurlyeq\zero$, and using that $\Identity_{2^n}\notin\PauliSet^{(n)}$ by definition, we have $\Tr\ofAlt{\rho_{x|P}}=1$ for all~$x$ and~$P$, so this is a family of valid quantum states.

For each Pauli string $P$, let $\eta(P)\in\set{0,1}$ be defined by $P^\top=(-1)^{\eta(P)}P$.
The POVM obtained from $\rho_{x|P}$ through \cref{eq:def_of_Axk_in_terms_of_rho} is then $\set{A_{x\oplus\eta(P)|P}}_x$.
Since the key~$P$ is revealed before Bob and Charlie measure, this key-dependent relabelling does not change the optimal winning probability, and below we identify these two families after this relabelling.

\begin{remark}
	\label{rem:L-without-the-identity}
    Although $\rho_{x|P}$ is not a valid quantum state when $P=\Identity_{2^n}$, we have that $\set[\big]{A_{x|P}}_x$ remains a valid POVM even when $P=\Identity_{2^n}$.
    Therefore, when computing upper bounds on $A_{x|P}$, we can consider sets $\PauliSet^{(n)} \subseteq \set{\Identity_2,X,Y,Z}^{\otimes n}$ containing the identity, which simplifies computations, and then deduce the result for $\PauliSet^{(n)}\backslash\set{\Identity_{2^n}}$ using the affine transformation from \Cref{prop:Pwin-of-P_n-union-Q_n}.
\end{remark}

        \subsection[Examples]{Examples of Pauli Encodings}
        \label{subsec:Examples-of-Pauli-Encodings}

In this manuscript, we study several variants of Pauli Encodings.

\begin{example}[BB84 Encoding $\KBBeightyfour$]
    \label{ex:BB84_Encoding} 
    The BB84 protocol is the first quantum key distribution (QKD) protocol and was introduced by \citeauthor{Bennett-Brassard-84} in 1984 \cite{Bennett-Brassard-84}.
    The idea is to encode a classical bit $x\in\set{0,1}$ with a key $\theta\in\set{0,1}$ into the qubit $H^\theta\ket{x}$,
    where $H=\frac{1}{\sqrt{2}}\begin{bsmallmatrix}
        1 & 1 \\
        1 & -1
    \end{bsmallmatrix}$ is the Hadamard gate.
    This ciphertext can be perfectly decoded using the correct key, and it yields uniformly random outcomes if using the incorrect key.
    This protocol is a Pauli Encoding for $n=1$ with $\KBBeightyfour^{(1)} = \set{X,Z}$ because the state $H^\theta\ketbra{x}{x}H^\theta$ can be rewritten as $\frac{1}{2}\of[\big]{\Identity_2+(-1)^xH^\theta Z H^\theta}$ and because $HZH = X$.
    This can be generalized to any number of qubits $n$ by taking: 
    $$
        \KBBeightyfour^{(n)} 
        \,\coloneqq\,
        \set{X,Z}^{\otimes n} 
        \,=\, 
        \set[\Big]{H^\theta (Z\otimes\cdots\otimes Z) H^\theta\,:\,\theta\in\set{0,1}^n}\,, 
    $$
    that we call the \emph{BB84 encoding},
    where $H^\theta\coloneqq H^{\theta_1}\otimes\cdots\otimes H^{\theta_n}$.
    Although this protocol is quite natural, we prove in \Cref{prop:Non_Security_of_the_BB84_Encoding} that it achieves neither strong nor weak security.
\end{example}

\begin{example}[HZH Encoding $\KHZH$]
    \label{ex:HZH_Encoding}
    One way to generalize the $\KBBeightyfour$ Encoding is to add variations on the $Z$ gates. 
    We introduce the \emph{HZH Encoding} by considering the following multiset:
    $$
        \KHZH^{(n,L)} 
        \,\coloneqq\, 
        \set[\Big]{
            H^\theta Z^a H^\theta\,:\,
            \theta\in\set{0,1}^n,\,
            a\in L
        }\,,
    $$
    where~$L$ is a subset of $\set{0,1}^n\backslash\set{\zero}$.
    In this paper, we study two specific examples of subsets~$L$. 
    The first one is $\Lall(n)\coloneqq\set{0,1}^n\backslash\set{\zero}$ and it yields the multiset obtained from $\set[\big]{\Identity_2,\Identity_2,X,Z}^{\otimes n}$ by removing its $2^n$ copies of $\Identity_{2^n}$.
    The second one is $\Lsym(w,n)\coloneqq\set[\big]{a\in\set{0,1}^n:\abs{a}=w}$, where $\abs{a}$ denotes the Hamming weight of~$a$.
    The Pauli Encodings associated with $\Lall(n)$ and $\Lsym(\lceil n/2\rceil,n)$ are promising in terms of security, see the discussion in \Cref{subsec:Security_Analysis_of_some_Efficient_Protocols}.
\end{example}

\begin{example}[Anticommuting Encoding $\Kanticomm$]
    \label{ex:Anticommuting_Encoding}
    The anticommuting Encoding was introduced and studied in \cite{Botteron-Broadbent-Culf-Nechita-Pellegrini-Rochette-24}.
    It is defined as follows:
    $$
        \Kanticomm^{(n)}
        \,\coloneqq\,
        \set[\Big]{ X^{\otimes(i-1)} \otimes Y \otimes \Identity_2^{\otimes(n-i)}\,:\,i\in\set{1,\ldots,n} }
        \,\cup\,
        \set[\Big]{ X^{\otimes(i-1)} \otimes Z \otimes \Identity_2^{\otimes(n-i)}\,:\,i\in\set{1,\ldots,n} }
        \,\cup\,
        \set{X^{\otimes n}}\,.
    $$
    It satisfies the property that any two distinct elements $P,Q\in\Kanticomm$ are anticommuting, \ie they satisfy $PQ=-QP$.
    The size of this set is $2n+1$, which makes this protocol inefficient.
    Note that, for a fixed number of qubits~$n$, it is possible to show that $\Kanticomm^{(n)}$ has the largest possible size among sets containing pairwise anticommuting elements only.
    We improve the security results from \cite{Botteron-Broadbent-Culf-Nechita-Pellegrini-Rochette-24} in \Cref{subsec:Security_Analysis_of_the_Anticommuting_Protocol}.
\end{example}

\begin{example}[Other Encodings]
    Here are other examples of Pauli Encodings that we study in this manuscript:
    $$ 
        \Kall^{(n)}
        \,\coloneqq\,
        \set[\big]{\Identity_2,X,Y,Z}^{\otimes n}\backslash\set{\Identity_{2^n}}
        \qquad\text{and}\qquad
        \Kreal^{(n)}
        \,\coloneqq\,
        \set[\big]{\Identity_2,X,Z}^{\otimes n}\backslash\set{\Identity_{2^n}}\,.
    $$
    These two Pauli Encodings behave well in terms of security, as discussed in \Cref{subsec:Security_Analysis_of_some_Efficient_Protocols}.
\end{example}

\subsection[Correctness]{Correctness of Pauli Encodings}

We prove that, given a key~$P\in\PauliSet$ and a quantum ciphertext~$\rho_{x|P}$, one can perfectly retrieve~$x$:

\begin{lemma}[Correctness]
    \label{lem:Correctness}
    The Pauli Encoding defined in \cref{eq:Pauli_Unclonable_Bit}
    is correct for any $n\in\NN^*$ and any $\PauliSet^{(n)}\subseteq\set{\Identity_2,X,Y,Z}^{\otimes n}\backslash\set{\Identity_{2^n}}$.
\end{lemma}

\begin{proof}
    Consider the decryption POVM $D_{x|P}=A_{x|P}$.
    Then, we have:
    \begin{align*}
        \Tr\ofAlt[\Big]{ D_{x|P}\,\rho_{x|P} }
        &\,=\,
        \Tr\ofAlt[\Bigg]{ \of[\bigg]{\frac{\Identity_{2^n} + (-1)^x P}{2}} \of[\bigg]{\frac{\Identity_{2^n} + (-1)^x P}{2^n}} }
        \\
        &\,=\,
        \frac{1}{2^{n+1}}\of[\Bigg]{ \Tr\ofAlt[\big]{\Identity_{2^n}} + 2(-1)^x\Tr\ofAlt[\big]{P} + \Tr\ofAlt[\big]{P^2} }
        \\
        &\,=\,
        \frac{1}{2^{n+1}}\of[\Big]{ 2^n + 0 + 2^n }
        \,=\,
        1\,,
    \end{align*}
    using that $\Tr\ofAlt[\big]{P}=0$ for $P\neq\Identity_{2^n}$ and that $P^2=\Identity_{2^n}$.
\end{proof}

\subsection[Efficiency]{Efficiency of Pauli Encodings}

Implementing an $n$-qubit Pauli string is efficient in the parameter~$n$. 
Moreover, the encrypted message $\rho^{(n)}_{x|P}$ can be prepared by sampling uniformly from product eigenstates of $P$ with total eigenvalue $(-1)^x$.
Therefore, the preparation of the encrypted message $\rho^{(n)}_{x|P}$ (see \Cref{def:Pauli_Encoding_state_version})
and the implementation of the decryption PVM $A^{(n)}_{x|P}$ (see \Cref{def:Pauli_Encoding_POVM_version})
are both efficient in the parameter~$n$.

It only remains to check that the key-generation algorithm (\ie sampling $P$) is efficient. 
This depends on the choice of $\PauliSet^{(n)}$. 
Among the examples of Pauli Encodings given in \Cref{subsec:Examples-of-Pauli-Encodings}, all are efficient except the anticommuting one that requires an exponential number of qubits in the security parameter~$\lambda$.

\section{Basic Properties of Pauli Encodings}

\subsection{Exact Expression of the Winning Probability}

The optimal winning probability is expressed as follows:
\begin{align*}
	\Prob^*\of[\big]{\text{\normalfont win MoE}}
	&\,=\,
	\sup_{B,C,\sigma}
	\sum_{x\in\set{0,1}} 
	\frac{1}{K}
	\sum_{P\in\PauliSet}
	\Tr\ofAlt[\Big]{ \of[\big]{ A_{x|P}\otimes B_{x|P}\otimes C_{x|P} } \sigma_{\Alice\Bob\Charlie} }
	\\
	&\,=\,
	\sup_{B,C}
	\,\lambda_{\max}\of[\Bigg]{
		\frac{1}{K}
		\sum_{x\in\set{0,1}} 
		\sum_{P\in\PauliSet}
		\of[\big]{ A_{x|P}\otimes B_{x|P}\otimes C_{x|P} }
	}\,.
\end{align*}
In order to simplify the winning probability expression, we use the following lemma adapted from \cite{Botteron-Broadbent-Culf-Nechita-Pellegrini-Rochette-24}:

\begin{lemma}
	\label{lem:change_B_and_C_into_unitaries}
	Without loss of generality, we may assume that: 
	$$
		B_{x|P}
		\,=\,
		C_{x|P}
		\,=\,
		\frac{\Identity_\Bob + (-1)^x V_{P}}{2}\,,
	$$
	for some Hermitian unitary $V_{P}$ of some dimension $d_\Bob$, 
	and we may assume that $\sigma_{\Alice\Bob\Charlie}$ is symmetric in the registers~$\Bob$ and~$\Charlie$.
\end{lemma}

\begin{proof}
	Using Naimark's Dilation Theorem, we can assume that Bob and Charlie use projective measurements. 
	Moreover, as the measurements are dichotomic, we can always write $B_{x|P}$ as:
	$$
		B_{x|P}
		\,=\,
		\frac{\Identity_\Bob + (-1)^x V_{P}}{2}\,,
	$$
	for some Hermitian unitary~$V_{P}$ of some dimension~$d_\Bob$.
	Finally, one can show that if a winning probability $q\in[0,1]$ can be achieved using a strategy $\of[\big]{\sigma, \set{B_{x|P}}, \set{C_{x|P}}}$, 
	then the same winning probability $q$ can be achieved by symmetric adversaries $\of[\big]{\sigma'', \set{B'_{x|P}}, \set{C'_{x|P}}}$ with $B'_{x|P} = C'_{x|P}$ and $\sigma''$ symmetric in Bob's and Charlie's registers.
	To this end, we consider
	Bob's and Charlie's new Hilbert spaces to be $\Hilbert_{\Bob'} = \Hilbert_{\Charlie'} = \Hilbert_{\Bob} \oplus \Hilbert_{\Charlie}$,
	we define the quantum state $\sigma'_{\Alice\Bob'\Charlie'}\coloneqq\zero\oplus\sigma_{\Alice\Bob\Charlie}\oplus\zero\oplus\zero
			\in \Density\of{\Hilbert_{\Alice\Bob'\Charlie'}}$,
	where:
	$$
		\Hilbert_{\Alice\Bob'\Charlie'}
		= \Hilbert_{\Alice}\otimes \Hilbert_{\Bob'} \otimes \Hilbert_{\Charlie'}
		= \of[\Big]{\Hilbert_{\Alice}\otimes \Hilbert_{\Bob} \otimes \Hilbert_{\Bob}}
		\oplus \of[\Big]{\Hilbert_{\Alice}\otimes \Hilbert_{\Bob} \otimes \Hilbert_{\Charlie}}
		\oplus \of[\Big]{\Hilbert_{\Alice}\otimes \Hilbert_{\Charlie} \otimes \Hilbert_{\Bob}}
		\oplus \of[\Big]{\Hilbert_{\Alice}\otimes \Hilbert_{\Charlie} \otimes \Hilbert_{\Charlie}}
		\,,
	$$
	we take the new quantum state to be $\sigma''_{\Alice\Bob'\Charlie'}\coloneqq\frac{1}{2}\of[\big]{\sigma'_{\Alice\Bob'\Charlie'}+F_{\Bob'\Charlie'}\sigma'_{\Alice\Bob'\Charlie'} F_{\Bob'\Charlie'}}$
	where $F_{\Bob'\Charlie'}$ is the swap operator between registers $\Bob'$ and $\Charlie'$,
	and we consider the direct sum $B'_{x|P} = C'_{x|P} \coloneqq B_{x|P} \oplus C_{x|P}$.
	In such a setting, we have that $\sigma''$ is symmetric in Bob's and Charlie's registers:
	$$
		F_{\Bob'\Charlie'}\sigma''_{\Alice\Bob'\Charlie'} F_{\Bob'\Charlie'}
		\,=\,
		\sigma''_{\Alice\Bob'\Charlie'}\,,
	$$
	and that the winning probability matches the expected one:
	\begin{align*}
		\hspace{1cm}
		&
		\hspace{-1cm}
		\Prob\of[\Big]{\text{\normalfont $\of[\big]{\sigma''_{\Alice\Bob'\Charlie'}, \set{B'_{x|P}}, \set{C'_{x|P}}}$ wins MoE}
		}
		\\
		& \,=\,
		\frac{1}{K} 
		\sum_{x\in\set{0,1}}  
		\sum_{P\in\PauliSet}
		\Tr\ofAlt[\Big]{ \of[\big]{ A_{x|P} \otimes B'_{x|P}\otimes C'_{x|P} } \sigma''_{\Alice\Bob'\Charlie'} }
		\\
		&\,=\,
		\frac{1}{K} 
		\sum_{x\in\set{0,1}}  
		\sum_{P\in\PauliSet}
		\Tr\ofAlt[\bigg]{ 
			\of[\big]{ A_{x|P} \otimes B'_{x|P}\otimes C'_{x|P} } \frac{1}{2}\sigma'_{\Alice\Bob'\Charlie'}
			+ \of[\big]{ A_{x|P} \otimes C'_{x|P}\otimes B'_{x|P} }  \frac{1}{2}\sigma'_{\Alice\Bob'\Charlie'}
		}
		\\
		& \,=\,
		\frac{1}{K} 
		\sum_{x\in\set{0,1}}  
		\sum_{P\in\PauliSet}
		\Tr\ofAlt[\Big]{ \of[\big]{ A_{x|P} \otimes B_{x|P}\otimes C_{x|P} } \sigma_{\Alice\Bob\Charlie} }
		\\
		& \,=\,
		\Prob\of[\Big]{\text{\normalfont $\of[\big]{\sigma_{\Alice\Bob\Charlie}, \set{B_{x|P}}, \set{C_{x|P}}}$ wins MoE} 
		}
		\,=\,
		q
		\,.
	\end{align*}
	Hence the result.
\end{proof}

Therefore, we obtain:
\begin{align}
	&
	\Prob^*\of[\big]{ \text{\normalfont win MoE} }\nonumber\\
	&
    \,=\,
	\sup_{V}
	\,\lambda_{\max}\of[\Bigg]{
		\frac{1}{K}
		\sum_{x\in\set{0,1}} 
		\sum_{P\in\PauliSet}
		\of[\bigg]{\frac{1}{2}\,\Identity_\Alice + \frac{1}{2}\,(-1)^x P} 
		\otimes \of[\bigg]{ \frac{1}{2}\,\Identity_\Bob + \frac{1}{2}\,(-1)^x V_P } 
		\otimes \of[\bigg]{ \frac{1}{2}\,\Identity_\Charlie + \frac{1}{2}\,(-1)^x V_P } 
	}
	\nonumber
	\\
	&\,=\,
	\frac{1}{8}
	\sup_{V}
	\,\lambda_{\max}\Bigg(
		\frac{1}{K}
		\sum_{x\in\set{0,1}} 
		\sum_{P\in\PauliSet}
		\bigg[
			\Identity_{\Alice\Bob\Charlie}
			+ (-1)^x \of[\Big]{ 
				\Identity_{\Alice\Bob}\otimes V_P
				+ \Identity_{\Alice}\otimes V_P \otimes \Identity_\Charlie
				+ P \otimes \Identity_{\Bob\Charlie}
			}
	\nonumber
	\\
	&\hspace{3cm}
			+ \of[\Big]{ 
				P \otimes V_P \otimes \Identity_\Charlie
				+ P \otimes \Identity_{\Bob} \otimes V_P
				+ \Identity_{\Alice} \otimes V_P \otimes V_P
			}
			+ (-1)^x P \otimes V_P \otimes V_P
		\bigg]
	\Bigg)
	\nonumber
	\\
	&\,=\,
	\frac{1}{8}
	\sup_{V}
	\,\lambda_{\max}\of[\Bigg]{
		\frac{1}{K}
		\sum_{P\in\PauliSet}
		\ofAlt[\bigg]{
			2\,\Identity_{\Alice\Bob\Charlie} 
			+ 2\of[\Big]{ 
				P \otimes V_P \otimes \Identity_\Charlie
				+ P \otimes \Identity_{\Bob} \otimes V_P
				+ \Identity_{\Alice} \otimes V_P \otimes V_P
			}
		} 
	}
	\nonumber
	\\
	&\,=\,
	\label{eq:expression-of-Pwin}
	\frac{1}{4}
	+
	\frac{1}{4}
	\sup_{V}
	\,\lambda_{\max}\of[\big]{
		T_V
	}\,,
\end{align}
where $T_V\coloneqq\frac{1}{K}
		\sum_{P\in\PauliSet}
		\of[\big]{ 
			P \otimes V_P \otimes \Identity_\Charlie
			+ P \otimes \Identity_{\Bob} \otimes V_P
			+ \Identity_{\Alice} \otimes V_P \otimes V_P
		}$.

\subsection{Naive Upper Bound and the Curse of the 3/4}
\label{rem:Curse-of-the-3/4}

	
	One might be tempted to use the subadditivity of $\lambda_{\max}$ on Hermitian matrices to obtain the following upper bound:
	\begin{align*}
		\lambda_{\max}\of[\big]{T_V}
		&\,\leq\,
		\lambda_{\max}\of[\Bigg]{
			\frac{1}{K}
			\sum_{P\in\PauliSet}
			\of[\big]{
			P \otimes V_P \otimes \Identity_\Charlie
			+ P \otimes \Identity_{\Bob} \otimes V_P
			}
		}
		+
		\frac{1}{K}
		\sum_{P\in\PauliSet}
		\underbrace{
		\lambda_{\max}\of[\Bigg]{
			\Identity_{\Alice} \otimes V_P \otimes V_P
		}
		}_{=1}
		\\
		&\,=\,
		\lambda_{\max}\of[\Bigg]{
			\frac{1}{K}
			\sum_{P\in\PauliSet}
			P \otimes \of[\big]{
			V_P \otimes \Identity_\Charlie
			+ \Identity_{\Bob} \otimes V_P
			}
		}
		+
		1\,,
	\end{align*}
	where we used that $\Identity_\Alice\otimes V_P\otimes V_P$ is Hermitian and unitary so its spectrum is included in $\set{\pm1}$. 
	This gives the following upper bound on the winning probability:
	\begin{equation}  \label{eq:upper-bound-on-pwin-with-the-1/2}
		\Prob^*\of[\big]{ \text{\normalfont win MoE} }
		\,\leq\,
		\underbrace{
		\frac12 + \frac{1}{4K}\,
		\sup_{V}
		\lambda_{\max}\of[\Bigg]{
			\sum_{P\in\PauliSet}
			P \otimes \of[\big]{
			V_P \otimes \Identity_\Charlie
			+ \Identity_{\Bob} \otimes V_P
			}
		}}_{\eqqcolon\,\Prob^*\of{ \text{\normalfont win \MoEprime} }}
		\,.
	\end{equation}
	This quantity belongs to $[0,1]$ and can be seen as a success probability that we write $\Prob^*\of[\big]{ \text{\normalfont win \MoEprime} }$.
	The corresponding game, that we call the \emph{pairwise MoE game}, can be seen as a relaxation of the usual MoE game, in which we average the marginals
	$\Prob\of{x_\Alice=x_\Bob \,|\, \sigma_{\Alice\Bob\Charlie}}\coloneqq\frac{1}{K} \sum_{k} \sum_{x} 
	\Tr\ofAlt[\big]{ 
		\of[\big]{A_{x|k}
		\otimes B_{x|k}
		\otimes \Identity_\Charlie}\,
		\sigma_{\Alice\Bob\Charlie}
	}$ 
	and $\Prob\of{x_\Alice=x_\Charlie\,|\, \sigma_{\Alice\Bob\Charlie}}$ for the same quantum state $\sigma_{\Alice\Bob\Charlie}$ (one can easily check that the average of the two marginals yields the expression as in \cref{eq:upper-bound-on-pwin-with-the-1/2}).

	But we argue that this upper bound is too loose and leads to the ``curse of the $3/4$''.
	This result tells us about the intrinsic tripartite nature of the unclonable-indistinguishable security.

\begin{proposition}
	\label{prop:Curse-of-the-3_4}
	For any correct binary encryption scheme $\set{\rho_{x|k}}_{x,k}$ for which $\set{A_{x|k}}_x$ defined in \cref{eq:def_of_Axk_in_terms_of_rho} is a POVM for every~$k$, the marginals $\Prob\of[\big]{x_\Alice = x_\Bob}$ and $\Prob\of[\big]{x_\Alice = x_\Charlie}$ can be simultaneously bounded away from $1/2$.
	More precisely, there exists a strategy, consisting of a state $\sigma_{\Alice\Bob\Charlie}$ and local POVMs, that yields:
	$$
		\Prob\of[\big]{x_\Alice = x_\Bob\,\big|\,\sigma_{\Alice\Bob\Charlie}}
		\,=\,
		\Prob\of[\big]{x_\Alice = x_\Charlie\,\big|\,\sigma_{\Alice\Bob\Charlie}}
		\,=\,
		\frac{3}{4}\,.
	$$
	Consequently, we have for any such encryption scheme:
	$$
		\Prob^*\of[\Big]{ \text{\normalfont win \MoEprime}}
		\,\geq\,
		\frac{3}{4}\,.
	$$
\end{proposition}

\begin{proof}
	For any such encryption scheme $\rho_{x|k}$ over some Hilbert space $\Hilbert_\Alice$, there exists by \Cref{def:Correctness} a decryption POVM $\set{D_{x|k}}_x\subseteq\Bounded(\Hilbert_\Alice)$ for every~$k$ such that for all $x$ and $k$:
	$$
		\Tr\ofAlt[\Big]{ D_{x|k}\,\,\rho_{x|k} }
		\,=\,
		1\,.
	$$
	In particular, we have $\Tr\ofAlt[\big]{ D_{x|k}\,\Phi(\rho_{x|k}) }=1$ with $\Phi=\Id:\Hilbert_\Alice \to \Hilbert_{\Alice'}$,
	where $\Hilbert_{\Alice'} = \Hilbert_{\Alice}$.
	We apply the Choi identity stating that $\Tr\ofAlt[\big]{ D_{x|k}\,\Phi(\rho_{x|k}) }=\Tr\ofAlt[\big]{ \of[\big]{\rho_{x|k}^\top\otimes D_{x|k}}\,C_\Phi }$,
	where $C_\Phi = d_\Alice \ketbra{\Omega}{\Omega}_{\Alice\Alice'}$ is the Choi matrix of $\Phi$ (it is the unnormalized maximally entangled rank-one operator in this case), 
	where $\ket{\Omega}_{\Alice\Alice'}\coloneqq\frac{1}{\sqrt{d_\Alice}}\sum_{i=0}^{d_\Alice-1}|i\rangle_{\Alice}\otimes|i\rangle_{\Alice'}$.
	Therefore, we obtain:
	\begin{equation}  \label{eq:expression-of-Omega-A-tensor-D-Omega}
		1
		\,=\,
		\Tr\ofAlt[\Big]{ \of[\big]{\rho_{x|k}^\top\otimes D_{x|k}}\,C_\Phi }
		\,=\,
		\Tr\ofAlt[\Big]{ \of[\big]{\textstyle\frac{2}{d_\Alice}\,A_{x|k}\otimes D_{x|k}}\, d_\Alice \ketbra{\Omega}{\Omega}_{\Alice\Alice'} }
		\,=\,
		2\,\bra{\Omega} \of[\big]{A_{x|k}\otimes D_{x|k}} \ket{\Omega}\,,
	\end{equation}
	where we used the relation between $\rho_{x|k}$ and $A_{x|k}$ from \cref{eq:def_of_Axk_in_terms_of_rho} in the second-to-last equality.
	Now, we explicitly construct a strategy achieving the wanted probability values.
	Consider two copies of the spaces $\Hilbert_{\Alice'} = \Hilbert_{\Alice''} = \Hilbert_\Alice$ and $\Complex_\Bob = \Complex_\Charlie = \Complex$ to keep track of the parties.
	In the respective spaces $\Hilbert_\Bob \coloneqq \Hilbert_{\Alice'}\oplus\Complex_\Bob$ and $\Hilbert_\Charlie \coloneqq \Hilbert_{\Alice''}\oplus\Complex_\Charlie$,
	consider the following POVMs for Bob and Charlie:
	$$
		B_{x|k} 
		\,=\,
		C_{x|k} 
		\,=\,
		D_{x|k} \oplus \frac12\,,
	$$
	which are positive semidefinite and sum to the identities $\Identity_{\Alice'}\oplus1_\Bob$ and $\Identity_{\Alice''}\oplus1_\Charlie$ respectively.
	The product space $\Hilbert_{\Alice\Bob\Charlie}$ is of the following form:
	\begin{align*}
		\Hilbert_{\Alice\Bob\Charlie}
		&\,=\, 
		\Hilbert_{\Alice}\otimes \Hilbert_{\Bob} \otimes \Hilbert_{\Charlie}
		\,=\, 
		\Hilbert_{\Alice}
		\otimes \of[\Big]{\Hilbert_{\Alice'}\oplus\Complex_\Bob } 
		\otimes \of[\Big]{\Hilbert_{\Alice''}\oplus\Complex_\Charlie }
		\\
		&\,\cong\,
		\of[\Big]{\Hilbert_{\Alice}\otimes \Hilbert_{\Alice'} \otimes \Hilbert_{\Alice''}}
		\oplus \of[\Big]{\Hilbert_{\Alice}\otimes \Hilbert_{\Alice'} \otimes \Complex_\Charlie}
		\oplus \of[\Big]{\Hilbert_{\Alice}\otimes \Complex_\Bob \otimes \Hilbert_{\Alice''}}
		\oplus \of[\Big]{\Hilbert_{\Alice}\otimes \Complex_\Bob \otimes \Complex_\Charlie}
		\,,
	\end{align*}
	in which we consider the following quantum state:
	$$
		\sigma_{\Alice\Bob\Charlie}
		\,\coloneqq\,
		\frac{1}{2}\ofAlt[\Big]{
			\zero
			\oplus \of[\Big]{ \ketbra{\Omega}{\Omega}_{\Alice\Alice'}\otimes 1_\Charlie }
			\oplus \of[\Big]{ \ketbra{\Omega}{\Omega}_{\Alice\Alice''}\otimes 1_\Bob }
			\oplus \zero
		}\,.
	$$
	Then, we have:
	\begin{align*}
		\Prob\of[\big]{x_\Alice = x_\Bob}
		&\,=\,
		\frac{1}{K} \sum_{k} \sum_{x} 
		\Tr\ofAlt[\Big]{ 
			\of[\big]{A_{x|k}
			\otimes B_{x|k}
			\otimes \Identity_\Charlie}\,
			\sigma_{\Alice\Bob\Charlie} 
		}
		\\
		&\,=\,
		\frac{1}{K} \sum_{k} \sum_{x} 
		\Tr\ofAlt[\Big]{ 
			\of[\Big]{A_{x|k}
			\otimes \of[\big]{D_{x|k} \oplus (1/2)_\Bob}
			\otimes \of[\big]{\Identity_{\Alice''} \oplus 1_\Charlie}}\,
			\sigma_{\Alice\Bob\Charlie} 
		}
		\\
		&\,=\,
		\frac{1}{K} \sum_{k} 
		\frac{1}{2}
		\Bigg(
		\sum_x
		\Tr\ofAlt[\Big]{ 
			\of[\big]{A_{x|k}
			\otimes D_{x|k} 
			\otimes 1_\Charlie}\,
			\of[\big]{\ketbra{\Omega}{\Omega}_{\Alice\Alice'}\otimes 1_\Charlie}
		}
		\\
		&
		\hspace{2.5cm} 
		+
		\sum_{x} 
		\Tr\ofAlt[\Big]{ 
			\of[\big]{A_{x|k}
			\otimes (1/2)_\Bob 
			\otimes \Identity_{\Alice''}}\,
			\of[\big]{\ketbra{\Omega}{\Omega}_{\Alice\Alice''}\otimes 1_\Bob}
		}
		\Bigg)\,.
	\end{align*}
	Using \cref{eq:expression-of-Omega-A-tensor-D-Omega},
	we get that the first trace is equal to $\bra{\Omega} \of[\big]{A_{x|k}\otimes D_{x|k}} \ket{\Omega} \times \Tr\ofAlt{1_\Charlie} = 1/2$, 
	so the first sum over~$x\in\set{0,1}$ is equal to~$1$.
	As for the second sum, 
	using the POVM normalization property $\sum_x A_{x|k} = \Identity_\Alice$ for any $k$, 
	we have:
	$$
		\sum_x \bra{\Omega} \of[\big]{A_{x|k}\otimes \Identity_{\Alice''}} \ket{\Omega} \times \Tr\ofAlt[\big]{(\frac12)_\Bob \times 1_\Bob}
		\,=\, 
		\bra{\Omega} \of[\big]{\Identity_{\Alice}\otimes \Identity_{\Alice''}} \ket{\Omega} \times \frac{1}{2}
		\,=\, 
		\frac{1}{2}\,.
	$$
	Finally, we obtain:
	$$
		\Prob\of[\big]{x_\Alice = x_\Bob}
		\,=\,
		\frac{1}{K} \sum_{k} 
		\frac{1}{2}
		\of[\Big]{ 1 + \frac{1}{2} }
		\,=\,
		\frac{3}{4}\,,
	$$
	and similarly for $\Prob\of[\big]{x_\Alice = x_\Charlie}$, as wanted.
\end{proof}

\subsection{Universal Lower Bound}
\label{subsec:Universal_Lower_Bound}

For all Pauli Encodings, an attack that perfectly targets one key and guesses otherwise has winning probability at least $1/2 + 1/(2K)$, yielding a trivial universal lower bound on the optimal winning probability.
Here, we improve the exponent of this universal bound by a factor~$2$.
Note that other universal lower bounds exist in the literature, especially in terms of the dimension $d_\Alice$ of the ciphertext, such as $\frac12+\Omega(\frac1{d_\Alice})$ \cite{Majenz-Schaffner-Tahmasbi-21} and its improvement $\frac12+\Omega(\frac1{\sqrt{d_\Alice}})$ \cite{Broadbent-Culf-Rochette-25}.


\begin{proposition}[Universal Lower Bound]
	\label{cor:Universal-Lower-Bound}
	For any $n\in\NN^*$ and any non-empty multiset $\PauliSet^{(n)}$ of elements of $\set{\Identity_2,X,Y,Z}^{\otimes n}$,
    the optimal winning probability is lower-bounded by:
    $$
		\Prob^*\of[\big]{ \text{\normalfont win MoE} }
        \,\geq\,
        \frac{1}{2} + \frac{1}{2\sqrt{K^{(n)}}}
        \,.
    $$
\end{proposition}

\begin{proof}
	It is sufficient to show that a value of at least $1/2 + 1/(2\sqrt{K})$ is achieved by one attack. Here, we consider the best attack between $V_P=\Identity$ and $V_P=-\Identity$.
	We have:
    \begin{align}
        \max_{s\in\set{0,1}} \Prob\of[\Big]{ \text{\normalfont $V=(-1)^s\Identity$ wins MoE} } \hspace{-2cm} & \nonumber\\
        &
        \,=\,
		\nonumber
        \frac{1}{4}
        +
        \frac{1}{4}
		\max_{s\in\set{0,1}}
        \lambda_{\max}\of[\Bigg]{
            \frac{1}{K}
            \sum_{P\in\PauliSet}
            \of[\Big]{ 
                (-1)^s\,P \otimes \Identity_{\Bob\Charlie}
                + (-1)^s\,P \otimes \Identity_{\Bob\Charlie}
                + \Identity_{\Alice\Bob\Charlie}
            }
        }
        \\
        &\,=\,
		\nonumber
        \frac{1}{2}
        +
        \frac{1}{2K}
		\max_{s\in\set{0,1}}
        \lambda_{\max}\of[\Bigg]{
            (-1)^s\sum_{P\in\PauliSet} P
		}
		\\
		&\,=\,
		\label{eq:max_s_P_V_I_wins}
		\frac{1}{2}
        +
        \frac{1}{2K}
        \,\operatornorm[\Bigg]{
            \,\sum_{P\in\PauliSet} P\,\,
        }\,,
    \end{align}
	where we used the hermiticity of $\sum_{P} P$ in the last line.
	Denote by $\lambda_{1},\dots,\lambda_{2^n}$ the eigenvalues of $\sum_{P} P$.
	Group the equal Pauli strings and denote by $m_Q$ the multiplicity of each distinct string~$Q$.
	Using the Hilbert--Schmidt orthogonality of distinct Pauli strings, we have:
    $$
		\sum_{i=1}^{2^n} \lambda_i^2
		\,=\,
        \Tr\ofAlt[\Bigg]{ \of[\bigg]{ \sum_{Q} m_Q Q }^2 }
        \,=\,
        2^n\sum_Q m_Q^2
        \,\geq\,
        2^n\sum_Q m_Q
        \,=\,
        2^n\,K\,.
    $$
    Hence, at least one of the $2^n$ eigenvalues satisfies $\lambda_i^2 \geq K$, which implies that\footnote{In the particular case of the anticommuting set $\PauliSet$, this inequality is tight, see \Cref{ex:P_win_for_anticommuting}.} $\operatornorm[\big]{\sum_{P} P}\geq \sqrt{K}$,
	and we obtain the wanted inequality.
\end{proof}

\subsection{Improved Lower Bounds}
\label{subsec:Improved-Lower-Bounds}

We compute the exact value of the winning probability for the attack $V_P=\pm\Identity$ or $V_P=P$ in specific choices of $\PauliSet$ in order to obtain better lower bounds than the universal one from \Cref{cor:Universal-Lower-Bound}.
The results are summarized in the table below.

\begin{example}[Anticommuting Case]
	\label{ex:P_win_for_anticommuting}
	If all pairs $(P,Q)$ of distinct elements in $\PauliSet$ are anticommuting (\ie that $PQ = -QP$),
	then $\sum_{P\neq Q}PQ = \zero$. 
	As a consequence, 
	$
		\of[\big]{ \sum_{P} P }^2
        \,=\,
        \sum_{P} P^2 + \sum_{P\neq Q} P\,Q 
        \,=\,
        K\,\,\Identity
	$
	and all the eigenvalues $\lambda_i$ of $\sum_P P$ satisfy $\lambda_i^2 = K$.
        Hence, $\operatornorm{\sum_{P} P}= \sqrt{K}$ and we obtain from \cref{eq:max_s_P_V_I_wins} that:
	$$
		\max_{s\in\set{0,1}}
		\Prob\of[\Big]{ \text{\normalfont $V_P=(-1)^s\Identity$ wins MoE} \,\Big|\,\text{$\PauliSet^{(n)}$ anticommuting} }
		\,=\,
		\frac{1}{2} + \frac{1}{2\sqrt{K^{(n)}}}
		\,,
	$$
	which coincides with the lower bound from \Cref{cor:Universal-Lower-Bound}.
\end{example}

\begin{example}[Commuting Case]
	\label{ex:P_win_for_commuting}
	If all pairs of elements in $\PauliSet$ are commuting (\ie if $PQ = QP$ for all $P,Q\in \PauliSet$), 
	then they are all diagonalizable in a common basis.
	We know that the eigenvalues of a Pauli string are $\pm1$.
	Suppose that all the Pauli strings $P\in \PauliSet$ admit a common $+1$ eigenvalue in this common basis (for instance, consider $\PauliSet^{(n)} = \set{\Identity_2,Z}^{\otimes n}$).
	Then
	$
		\lambda_{\max}\of[\big]{
			\,\sum_{P} P\,\,
		}
		\,=\,
		\sum_{P\in\PauliSet} 1
		\,=\,
		K
	$.
	Thus, from \cref{eq:max_s_P_V_I_wins}, we obtain:
	$$
		\Prob\of[\Big]{ \text{\normalfont $V_P=\Identity$ wins MoE} \,\Big|\,\text{$\PauliSet^{(n)}$ commuting with common $+1$ eigenvalue} }
		\,=\,
		1
		\,,
	$$
	and the extreme two cases of
	$
		1/2 + 1/(2\sqrt{K})
		\leq
		\max_s
		\Prob\of[\big]{ \text{\normalfont $V=(-1)^s\Identity$ wins MoE} }
		\leq
		1
	$
	are achieved by the anticommuting and the commuting cases respectively.
\end{example}

\begin{example}[All Pauli Strings]
	\label{ex:P_win_for_all_Pauli_strings}
	Consider $\PauliSet^{(n)}=\set{\Identity_2,X,Y,Z}^{\otimes n}$. Then, we have:
	\begin{equation}
		\label{eq:sum_of_P}
		\sum_{P\in\PauliSet^{(n)}} P
		\,=\,
		\sum_{P_1,\ldots,P_n\in\set{\Identity_2,X,Y,Z}}\,\bigotimes_{i=1}^n P_i
		\,=\,
		\bigotimes_{i=1}^n \of[\Bigg]{\sum_{P_i\in\set{\Identity_2,X,Y,Z}} P_i}
		\,=\,
		\begin{bmatrix}
			2 & 1-i \\
			1+i & 0
		\end{bmatrix}^{\otimes n}\,.
	\end{equation}
	So when $V_P=\Identity$, we have:
	\begin{align*}
		T_V
		&
		\,=\,
		\frac{2}{4^n}
		\of[\Bigg]{\sum_{P\in\PauliSet} P} \otimes \Identity_{\Bob\Charlie}
		+ \Identity_{\Alice\Bob\Charlie}
		\,=\,
		\frac{2}{4^n}
		\begin{bmatrix}
			2 & 1-i \\
			1+i & 0
		\end{bmatrix}^{\otimes n} \otimes \Identity_{\Bob\Charlie}
		+ \Identity_{\Alice\Bob\Charlie}\,.
	\end{align*}
	Therefore, the maximal eigenvalue is:
	\begin{align*}
		\lambda_{\max}\of[\big]{ T_V }
		\,=\,
		\frac{2}{4^n}\,
		\lambda_{\max}\of[\Bigg]{
		\begin{bmatrix}
			2 & 1-i \\
			1+i & 0
		\end{bmatrix}}^{n} 
		+ 1
		\,=\,
		2\, \of[\bigg]{ \frac{1+\sqrt{3}}{4} }^n +1\,,
	\end{align*}
	and the winning probability is:
	$$
		\Prob\of[\Big]{ \text{\normalfont $V_P=\Identity$ wins MoE}\,\Big|\,\PauliSet^{(n)} = \set{\Identity_2,X,Y,Z}^{\otimes n} }
		\,=\,
		\frac{1}{2} + \frac{1}{2}\,\of[\bigg]{ \frac{1+\sqrt{3}}{4} }^n
		\,.
	$$
	Let us compute another example, when $V_P=P$. 
	With similar computations as in \cref{eq:sum_of_P}, we obtain that 
	$\sum_{P\in\PauliSet} P\otimes P 
		\,=\, 
		2^n\, F^{\otimes n}$ where $F =
	\begin{bsmallmatrix}
		1 & 0 & 0 & 0 \\
		0 & 0 & 1 & 0 \\
		0 & 1 & 0 & 0 \\
		0 & 0 & 0 & 1 \\
	\end{bsmallmatrix}$ 
	is the flip operator (also known as the SWAP operator).
	As a consequence:
	$$
		T_V
		\,=\,
		\frac{1}{4^n}
		\sum_{P\in\PauliSet}
		\of[\Big]{ 
			P \otimes P \otimes \Identity_\Charlie
			+ P \otimes \Identity_{\Bob} \otimes P
			+ \Identity_{\Alice} \otimes P \otimes P
		}
		\,=\,
		\frac{1}{2^n}
		\of[\bigg]{
			F_{\Alice\Bob}^{\otimes n}\otimes \Identity_\Charlie
			+ F_{\Alice\Charlie}^{\otimes n}\otimes \Identity_\Bob
			+ \Identity_\Alice \otimes F_{\Bob\Charlie}^{\otimes n}
		}\,.
	$$
	On the one hand, by subadditivity of the maximal eigenvalue, one can see that:
	$$
		\lambda_{\max}(T_V) 
		\,\leq\, 
		\frac{3}{2^n}\,{\lambda_{\max}\of[\big]{F^{\otimes n}\otimes \Identity} } 
		\,=\, 
		\frac{3}{2^n}\,.
	$$
	On the other hand, one can observe that the value $3/2^n$ belongs to the spectrum of $T_V$: for instance, consider the eigenvector $\ket{1}^{\otimes 3n}$. 
	Hence $\lambda_{\max}(T_V) = 3/2^n$ and we have:
	$$
		\Prob\of[\Big]{ \text{\normalfont $V_P=P$ wins MoE}\,\Big|\,\PauliSet^{(n)} = \set{\Identity_2,X,Y,Z}^{\otimes n} }
		\,=\,
		\frac{1}{4} + \frac{3}{2^{n+2}}
		\,.
	$$
\end{example}

In a similar way as in \Cref{ex:P_win_for_all_Pauli_strings}, one can compute the winning probability of the attack $V_P=\Identity$ for $\set{\Identity_2,X,Z}^{\otimes n}$ and for the augmented HZH multiset $\set[\big]{\Identity_2,\Identity_2,X,Z}^{\otimes n}$, yielding the result displayed in the table below.
Moreover, for $\set{X,Z}^{\otimes n}$, we know from \Cref{prop:Non_Security_of_the_BB84_Encoding} that there is a much better attack in this case, achieving the constant value $1/2+1/(2\sqrt{2})$ for all~$n$.
Hence, by collecting all the results,
we obtain the following expressions:
$$
	\begin{array}{ccc}
		\text{\bf Pauli Set $\PauliSet$} 
		& \text{\bf Universal Lower Bound}
		& \text{\bf Improved Lower Bound}
		\\
		\hline
		\text{Commuting} & \frac{1}{2} + \frac{1}{2\sqrt{K}} & 1
		\\
		\set{X,Z}^{\otimes n} & \frac{1}{2} + \frac{1}{2\sqrt{2}^n} & \frac{1}{2} + \frac{1}{2\sqrt{2}}
		\\
		\text{Anticommuting} & \frac{1}{2} + \frac{1}{2\sqrt{K}} & \frac{1}{2} + \frac{1}{2\sqrt{K}}
		\\
		\set[\big]{\Identity_2,\Identity_2,X,Z}^{\otimes n} & \frac{1}{2} + \frac{1}{2\times2^n} & \frac{1}{2} + \frac{1}{2}\,\of[\big]{ \frac12 + \frac{1}{2\sqrt{2}} }^n
		\\
		\set{\Identity_2,X,Z}^{\otimes n} & \frac{1}{2} + \frac{1}{2\sqrt{3}^n} & \frac{1}{2} + \frac{1}{2}\,\of[\big]{ \frac{1+\sqrt{2}}{3} }^n
		\\
		\set{\Identity_2,X,Y,Z}^{\otimes n} & \frac{1}{2} + \frac{1}{2\times2^n} & \frac{1}{2} + \frac{1}{2}\,\of[\big]{ \frac{1+\sqrt{3}}{4} }^n
	\end{array}
$$
As discussed in \Cref{subsec:Security_Analysis_of_some_Efficient_Protocols}, none of our numerical simulations allowed us to obtain higher lower bounds than those presented in the right column, giving the intuition that they are optimal.

\subsection{Union of Two Pauli Sets}

Here is the relation between the optimal winning probability for $\PauliSet^{(n)}\sqcup\MPauliSet^{(n)}$ and the one for $\PauliSet^{(n)}$:

\begin{lemma}
	\label{prop:Pwin-of-P_n-union-Q_n}
	Consider an integer $n\in\NN^*$ and two disjoint multisets $\PauliSet^{(n)},\MPauliSet^{(n)} \subseteq \set{\Identity_2,X,Y,Z}^{\otimes n}$. 
	Assume that $\Tr\ofAlt[\Big]{ \of[\big]{\sum_{P\in\MPauliSet^{(n)}}P\otimes \Identity_{\Bob\Charlie}}\sigma_{\Alice\Bob\Charlie} } = |\MPauliSet^{(n)}|$ for any quantum state $\sigma_{\Alice\Bob\Charlie}$.
	Then:
	\begin{align*}
		\,\,\,\Prob^*\of[\Big]{
			\text{\normalfont win MoE} 
			\,\Big| 
			\PauliSet^{(n)}\sqcup\MPauliSet^{(n)}
		}
		\,=\,
		\frac{\abs{\PauliSet^{(n)}}}{\abs{\PauliSet^{(n)}} + \abs{\MPauliSet^{(n)}}}\,
		\Prob^*\of[\Big]{
			\text{\normalfont win MoE} 
			\,\Big| 
			\PauliSet^{(n)}
		}
		+
		\frac{\abs{\MPauliSet^{(n)}}}{\abs{\PauliSet^{(n)}} + \abs{\MPauliSet^{(n)}}}\,.
	\end{align*}
	In particular, taking $\MPauliSet^{(n)}=\set[\big]{\Identity_{2^n}}$, we have:
	$$
		\Prob^*\of[\Big]{
			\text{\normalfont win MoE} 
			\,\Big| 
			\PauliSet^{(n)}\sqcup\set[\big]{\Identity_{2^n}}
		}
		\,=\,
		\frac{\abs{\PauliSet^{(n)}}}{\abs{\PauliSet^{(n)}} + 1}\,
		\Prob^*\of[\Big]{
			\text{\normalfont win MoE} 
			\,\Big| 
			\PauliSet^{(n)}
		}
		+
		\frac{1}{\abs{\PauliSet^{(n)}} + 1}
		\,.
	$$
\end{lemma}

\begin{proof}
	We show the equality by proving inequalities in both ways.
	On the one hand, consider some strategy $\of[\big]{\sigma_{\Alice\Bob\Charlie}, \set{B_{x|P}}, \set{C_{x|P}}}$ for the MoE game associated with~$\PauliSet$. 
	We construct the following strategy for the MoE game associated with $\PauliSet\sqcup\MPauliSet$:
	\begin{align*}
		\sigma_{\Alice\Bob\Charlie}'
		\,\coloneqq\,& 
		\sigma_{\Alice\Bob\Charlie}
		\,,\qquad
		B_{x|P}'
		\,\coloneqq\,
		\left\{
		\begin{array}{cl}
			B_{x|P} & \text{if $P\in\PauliSet$,} \\
			\delta_{x=0}\,\Identity_\Bob & \text{if $P\in\MPauliSet$,}
		\end{array}
		\right.
		\,,\qquad
		C_{x|P}'
		\,\coloneqq\,
		\left\{
		\begin{array}{cl}
			C_{x|P} & \text{if $P\in\PauliSet$,} \\
			\delta_{x=0}\,\Identity_\Charlie & \text{if $P\in\MPauliSet$,}
		\end{array}
		\right.
	\end{align*}
	for any $P\in \PauliSet\sqcup\MPauliSet$.
	They form a valid state and valid POVMs.
	We have:
	\begin{align*}
		\Prob^*\of[\Big]{
			\text{\normalfont win MoE} 
			\,\Big| 
			\PauliSet\sqcup\MPauliSet
		}
		\,\geq\,\,&
		\Prob\of[\Big]{
			\text{\normalfont $\of[\big]{\sigma'_{\Alice\Bob\Charlie}, \set{B'_{x|P}}, \set{C'_{x|P}}}$ wins MoE} 
			\,\Big| 
			\PauliSet\sqcup\MPauliSet
		}
		\\
		\,=\,&
		\frac{1}{\abs[\big]{\PauliSet\sqcup\MPauliSet}} 
		\sum_{x\in\set{0,1}} 
		\sum_{P\in\PauliSet\sqcup\MPauliSet} 
		\Tr\ofAlt[\Bigg]{\of[\Big]{
			A_{x|P}^{(n)}
			\otimes {B'}_{x|P} 
			\otimes {C'}_{x|P}
		} 
		\sigma'_{\Alice\Bob\Charlie} }
		\\
		\,=\,&
		\frac{\abs[\big]{\PauliSet}}{\abs[\big]{\PauliSet\sqcup\MPauliSet}} 
		\times
		\frac{1}{K} 
		\sum_{x\in\set{0,1}} 
		\sum_{P\in\PauliSet} 
		\Tr\ofAlt[\Bigg]{\of[\Big]{
			A_{x|P}^{(n)}
			\otimes {B}_{x|P} 
			\otimes {C}_{x|P}
		} 
		\sigma_{\Alice\Bob\Charlie} }
		\\
		&
		+\,
		\frac{1}{\abs[\big]{\PauliSet\sqcup\MPauliSet}} 
		\sum_{x\in\set{0,1}} 
		\sum_{P\in\MPauliSet} 
		\Tr\ofAlt[\Bigg]{
		\delta_{x=0}
		\of[\Big]{
			A_{x|P}^{(n)}
			\otimes \Identity_\Bob
			\otimes \Identity_\Charlie
		} 
		\sigma_{\Alice\Bob\Charlie} }\,.
	\end{align*}
	In the second-to-last line, we recognize the expression of the winning probability of $\of[\big]{\sigma_{\Alice\Bob\Charlie}, \set{B_{x|P}}, \set{C_{x|P}}}$ for the $\PauliSet$-MoE game.
	Let us compute the last sum. 
	The summation is non-zero only when $x=0$. 
	It only remains $\Tr\ofAlt[\Big]{ \of[\big]{\sum_{P\in\MPauliSet}A_{0|P}^{(n)}\otimes \Identity_{\Bob\Charlie}}\sigma_{\Alice\Bob\Charlie} }$, which is equal to $\abs{\MPauliSet}$ by assumption.
	Therefore, we have:
	$$
		\Prob^*\of[\Big]{
			\text{\normalfont win MoE} 
			\,\Big| 
			\PauliSet\sqcup\MPauliSet
		}
		\,\geq\,
		\frac{\abs[\big]{\PauliSet}}{\abs[\big]{\PauliSet\sqcup\MPauliSet}}\,\,
		\Prob\of[\Big]{
			\text{\normalfont $\of[\big]{\sigma_{\Alice\Bob\Charlie}, \set{B_{x|P}}, \set{C_{x|P}}}$ wins MoE} 
			\,\Big| 
			\PauliSet
		}
		+ \frac{\abs{\MPauliSet}}{\abs[\big]{\PauliSet\sqcup\MPauliSet}} 
		\,,
	$$
	and we conclude by taking the supremum over all possible strategies $\of[\big]{\sigma_{\Alice\Bob\Charlie}, \set{B_{x|P}}, \set{C_{x|P}}}$ for the $\PauliSet$-MoE game, obtaining the first inequality as wanted.

	On the other hand, consider some strategy $\of[\big]{\sigma_{\Alice\Bob\Charlie}, \set{B_{x|P}}, \set{C_{x|P}}}$ for the $\of[\big]{\PauliSet\sqcup\MPauliSet}$-MoE game.
	We construct the following strategy for the $\PauliSet$-MoE game:
	$$
		\sigma_{\Alice\Bob\Charlie}'
		\,\coloneqq\,
		\sigma_{\Alice\Bob\Charlie}
		\,,\quad
		B_{x|P}'
		\,\coloneqq\,
		B_{x|P}
		\,,\quad
		C_{x|P}'
		\,\coloneqq\,
		C_{x|P}\,,
	$$
	for any $P\in \PauliSet$.
	They form a valid state and valid POVMs.
	We have:
	\begin{align*}
		\Prob^*\of[\Big]{
			\text{\normalfont win MoE} 
			\,\Big| 
			\PauliSet
		}
		\,\geq\,\,&
		\Prob\of[\Big]{
			\text{\normalfont $\of[\big]{\sigma'_{\Alice\Bob\Charlie}, \set{B'_{x|P}}, \set{C'_{x|P}}}$ wins MoE} 
			\,\Big| 
			\PauliSet
		}
		\\
		\,=\,&
		\frac{1}{\abs[\big]{\PauliSet}} 
		\sum_{x\in\set{0,1}} 
		\sum_{P\in\PauliSet} 
		\Tr\ofAlt[\Bigg]{\of[\Big]{
			A_{x|P}^{(n)} 
			\otimes {B'}_{x|P} 
			\otimes {C'}_{x|P}
		} 
		\sigma'_{\Alice\Bob\Charlie} }
		\\
		\,=\,&
		\frac{\abs[\big]{\PauliSet\sqcup\MPauliSet}}{\abs[\big]{\PauliSet}}
		\times
		\frac{1}{\abs[\big]{\PauliSet\sqcup\MPauliSet}} 
		\sum_{x\in\set{0,1}} 
		\sum_{P\in\PauliSet\sqcup\MPauliSet}
		\Tr\ofAlt[\Bigg]{\of[\Big]{
			A_{x|P}^{(n)} 
			\otimes {B}_{x|P} 
			\otimes {C}_{x|P}
		} 
		\sigma_{\Alice\Bob\Charlie} }
		\\
		&
		-\,
		\frac{1}{\abs[\big]{\PauliSet}} 
		\sum_{x\in\set{0,1}} 
		\sum_{P\in\MPauliSet} 
		\Tr\ofAlt[\Bigg]{\of[\Big]{
			A_{x|P}^{(n)} 
			\otimes {B}_{x|P} 
			\otimes {C}_{x|P}
		} 
		\sigma_{\Alice\Bob\Charlie} }\,.
	\end{align*}
	In the second-to-last line, we recognize the expression of the winning probability of $\of[\big]{\sigma_{\Alice\Bob\Charlie}, \set{B_{x|P}}, \set{C_{x|P}}}$ for the $\of[\big]{\PauliSet\sqcup\MPauliSet}$-MoE game.
	As for the last line, using that $A_{x|P}^{(n)} \succeq\zero$, and ${B}_{x|P}\preceq\Identity_\Bob$, and ${C}_{x|P}\preceq\Identity_\Charlie$ for all $x$ and $P$, we can lower bound it by:
	$$
		-\,
		\frac{1}{\abs[\big]{\PauliSet}} 
		\sum_{x\in\set{0,1}} 
		\sum_{P\in\MPauliSet} 
		\Tr\ofAlt[\Bigg]{\of[\Big]{
			A_{x|P}^{(n)} 
			\otimes \Identity_\Bob
			\otimes \Identity_\Charlie
		} 
		\sigma_{\Alice\Bob\Charlie} }
		\,=\,
		-\frac{\abs[\big]{\MPauliSet}}{\abs[\big]{\PauliSet}} 
		\,.
	$$
	Therefore, we obtain:
	$$
		\Prob^*\of[\Big]{
			\text{\normalfont win MoE} 
			\,\Big| 
			\PauliSet
		}
		\,\geq\,
		\frac{\abs[\big]{\PauliSet\sqcup\MPauliSet}}{\abs[\big]{\PauliSet}}\,\,
		\Prob\of[\Big]{
			\text{\normalfont $\of[\big]{\sigma_{\Alice\Bob\Charlie}, \set{B_{x|P}}, \set{C_{x|P}}}$ wins MoE} 
			\,\Big| 
			\PauliSet\sqcup\MPauliSet
		}
		- \frac{\abs[\big]{\MPauliSet}}{\abs[\big]{\PauliSet}}
		\,,
	$$
	which gives exactly the other wanted inequality after taking the supremum over all possible strategies $\of[\big]{\sigma_{\Alice\Bob\Charlie}, \set{B_{x|P}}, \set{C_{x|P}}}$ for the $\of[\big]{\PauliSet\sqcup\MPauliSet}$-MoE game.
	Hence the equality.
\end{proof}

\section{Towards Unconditional Security}

\subsection{The BB84 Encoding is not Secure}

The BB84 Encoding is defined in \Cref{ex:BB84_Encoding} as a generalisation of the BB84 one-qubit protocol to any number $n$ of qubits.
It is precisely the ``XOR repetition'' monogamy-of-entanglement game studied in \cite[Theorem~3]{Coladangelo-Liu-Xie-25}, where the exact winning probability stated below was proved for every~$n$ and an optimal semi-classical strategy was exhibited.
We give below a self-contained proof in the Pauli-Encoding formalism.

\begin{proposition}[Non-Security of the BB84 Encoding]
	\label{prop:Non_Security_of_the_BB84_Encoding}
	The BB84 Encoding satisfies:
	$$
		\forall n\in\NN^*,\qquad
		\Prob^*\of[\Big]{
			\text{\normalfont win MoE}
			\,\Big|\,
			\KBBeightyfour^{(n)}
		}
		=
		\cos^2\of[\Big]{\frac{\pi}{8}}
		=
		\frac12+\frac{1}{2\sqrt{2}}\,.
	$$
	In particular, the BB84 Encoding is neither weakly nor strongly secure.
\end{proposition}

\begin{proof}
	Let
	$
		p_n
		\coloneqq
		\Prob^*\of[\big]{
			\text{\normalfont win MoE}
			\,\big|\,
			\KBBeightyfour^{(n)}
		}
	$.
	We first prove that this is a decreasing sequence, \ie that $p_{n+1}\leq p_n$.
	Recall the induction relation
	$
		\KBBeightyfour^{(n+1)}
		=
		\KBBeightyfour^{(n)}\otimes\set{X,Z}
	$.
	For every $P\in\KBBeightyfour^{(n)}$, $Q\in\set{X,Z}$, and $x\in\set{0,1}$, one has:
	\begin{equation}
		\label{eq:BB84-projector-recursion}
		A_{x|P\otimes Q}^{(n+1)}
		\,=\,
		\frac{\Identity+(-1)^xP\otimes Q}{2}
		\,=\,
		\frac14
		\sum_{b\in\set{0,1}}
		\of[\big]{\Identity+(-1)^{x\oplus b}P}
		\otimes
		\of[\big]{\Identity+(-1)^bQ}
		\,=\,
		\sum_{b\in\set{0,1}}
		A_{x\oplus b|P}^{(n)}
		\otimes
		A_{b|Q}^{(1)}\,.
	\end{equation}
	Now, consider any strategy 
	$
		\of[\big]{
			\sigma,
			\set{B_{x|P\otimes Q}},
			\set{C_{x|P\otimes Q}}
		}
	$
	for the $(n+1)$-qubit game.
	We construct a strategy for the $n$-qubit game as follows.
	Before the game, choose $Q\in\set{X,Z}$ uniformly, measure the last qubit of Alice with the projective measurement $\set{A_{b|Q}^{(1)}}_b$, and give the classical pair $(b,Q)$ to both Bob and Charlie.
	Upon receiving the key $P$, they use their original measurements for the key $P\otimes Q$ and flip their outputs by~$b$.
	Denote by $B'_{x'|P}$ and $C'_{x'|P}$ the resulting POVMs.
	Then:
	\begin{align*}
		&\hspace{-1cm}\Prob\of[\Big]{
			\text{\normalfont $\of[\big]{\sigma,\set{B'},\set{C'}}$ wins MoE}
			\,\Big|\,
			\KBBeightyfour^{(n)}
		}
		\\
		=\,&
		\frac{1}{2\,\abs{\KBBeightyfour^{(n)}}}
		\sum_{\substack{P\in\KBBeightyfour^{(n)}\\Q\in\set{X,Z}}}
		\sum_{x,b\in\set{0,1}}
		\Tr\ofAlt[\Big]{
			\of[\Big]{
				A_{x|P}^{(n)}
				\otimes
				A_{b|Q}^{(1)}
				\otimes
				B_{x\oplus b|P\otimes Q}
				\otimes
				C_{x\oplus b|P\otimes Q}
			}
			\sigma
		}
		\\
		=\,&
		\frac{1}{2\,\abs{\KBBeightyfour^{(n)}}}
		\sum_{\substack{P\in\KBBeightyfour^{(n)}\\Q\in\set{X,Z}}}
		\sum_{x'\in\set{0,1}}
		\Tr\ofAlt[\Bigg]{
			\of[\Bigg]{
				\sum_{b\in\set{0,1}}
				A_{x'\oplus b|P}^{(n)}
				\otimes
				A_{b|Q}^{(1)}
			}
			\otimes
			B_{x'|P\otimes Q}
			\otimes
			C_{x'|P\otimes Q}
			\,
			\sigma
		}
		\\
		=\,&
		\frac{1}{\abs{\KBBeightyfour^{(n+1)}}}
		\sum_{P\otimes Q\in\KBBeightyfour^{(n+1)}}
		\sum_{x'\in\set{0,1}}
		\Tr\ofAlt[\Big]{
			\of[\Big]{
				A_{x'|P\otimes Q}^{(n+1)}
				\otimes
				B_{x'|P\otimes Q}
				\otimes
				C_{x'|P\otimes Q}
			}
			\sigma
		}
		\\
		=\,&
		\Prob\of[\Big]{
			\text{\normalfont $\of[\big]{\sigma,\set{B},\set{C}}$ wins MoE}
			\,\Big|\,
			\KBBeightyfour^{(n+1)}
		}\,,
	\end{align*}
	where we used \cref{eq:BB84-projector-recursion}.
	Consequently, we have $p_{n+1}\leq p_n$.

	For $n=1$, the game is the usual BB84 monogamy-of-entanglement game like in \cite{Tomamichel-Fehr-Kaniewski-Wehner-13}.
	The standard two-question overlap bound gives:
	\[
		p_1
		\leq
		\frac12
		\of[\Big]{
			1+
			\max_{x,z\in\set{0,1}}
			\operatornorm[\big]{
				A_{x|X}^{(1)}A_{z|Z}^{(1)}
			}
		}
		=
		\frac12+\frac{1}{2\sqrt{2}}\,,
	\]
	and this value is attained by the Breidbart strategy.
	Hence, we obtain the first inequality:
	\begin{equation}
		\label{eq:BB84-upper-bound-all-n}
		\forall n\in\NN^*,\qquad
		p_n
		\leq
		p_1
		=
		\cos^2\of[\Big]{\frac{\pi}{8}}\,.
	\end{equation}

	It remains to prove the matching lower bound.
	For $z\in\set{0,1}^n$, define:
	\[
		\epsilon(z)
		\coloneqq
		z_1\oplus\cdots\oplus z_n,
		\qquad
		g(z)
		\coloneqq
		\bigoplus_{1\leq i<j\leq n}z_iz_j\,,
	\]
	and let:
	\[
		\tau
		\coloneqq
		\tan\of[\Big]{\frac{\pi}{8}}
		=
		\sqrt{2}-1\,.
	\]
	Consider the unit vector
	$$
		\ket{w_n}
		\coloneqq
		\frac{\cos\of[\big]{\pi/8}}{2^{(n-1)/2}}
		\sum_{z\in\set{0,1}^n}
		(-1)^{g(z)}
		\tau^{\epsilon(z)}
		\ket{z}\,.
	$$
	Its normalisation follows from the fact that exactly $2^{n-1}$ strings have each parity and
	$
		\cos^2\of[\big]{\frac{\pi}{8}}
		\of{1+\tau^2}
		=
		1.
	$
	For $\theta\in\set{0,1}^n$, write:
	\[
		P_\theta
		\coloneqq
		H^\theta Z^{\otimes n}H^\theta
		\in
		\KBBeightyfour^{(n)}\,.
	\]
	The map $\theta\mapsto P_\theta$ is a bijection.
	Moreover,
	\begin{equation}
		\label{eq:BB84-state-correlations}
		\bra{w_n}P_\theta\ket{w_n}
		=
		\frac{(-1)^{g(\theta)}}{\sqrt{2}}.
	\end{equation}
	To verify this identity, observe that
	$
		P_\theta\ket{z}
		=
		(-1)^{\bigoplus_{j:\theta_j=0}z_j}
		\ket{z\oplus\theta}
	$
	and that a direct expansion over $\mathbb{F}_2$ gives:
	\[
		g(z)
		\oplus
		g(z\oplus\theta)
		\oplus
		\bigoplus_{j:\theta_j=0}z_j
		=
		g(\theta)
		\oplus
		\of[\big]{1\oplus\epsilon(\theta)}\epsilon(z)\,.
	\]
	Therefore,
	\begin{align*}
		\bra{w_n}P_\theta\ket{w_n}
		&\,=\,
		\frac{\cos^2\of[\big]{\pi/8}}{2^{n-1}}
		(-1)^{g(\theta)}
		\sum_{z\in\set{0,1}^n}
		(-1)^{\of[\big]{1\oplus\epsilon(\theta)}\epsilon(z)}
		\tau^{\epsilon(z)+\epsilon(z\oplus\theta)}
		\\
		&\,=\,
		\frac{(-1)^{g(\theta)}\cos^2\of[\big]{\pi/8}}{2^{n-1}}
		\begin{cases}
			2^{n-1}\of[\big]{1-\tau^2},
			&\text{if }\epsilon(\theta)=0\,,
			\\[1mm]
			2^n\tau,
			&\text{if }\epsilon(\theta)=1\,.
		\end{cases}
	\end{align*}
	Both cases equal $(-1)^{g(\theta)}/\sqrt{2}$ because
	$
		\cos^2\of[\big]{\frac{\pi}{8}}\of[\big]{1-\tau^2}
		=
		2\cos^2\of[\big]{\frac{\pi}{8}}\tau
		=
		1/\sqrt{2}
	$,
	which proves \cref{eq:BB84-state-correlations}.

	Finally, Alice is given the state $\ketbra{w_n}{w_n}$, while Bob and Charlie use the deterministic measurements:
	\[
		B_{x|P_\theta}
		=
		C_{x|P_\theta}
		=
		\delta_{x,g(\theta)}.
	\]
	The winning probability of this strategy is:
	\begin{align*}
		\frac{1}{2^n}
		\sum_{\theta\in\set{0,1}^n}
		\bra{w_n}
		A_{g(\theta)|P_\theta}^{(n)}
		\ket{w_n}
		&=
		\frac{1}{2^n}
		\sum_{\theta\in\set{0,1}^n}
		\frac{
			1+
			(-1)^{g(\theta)}
			\bra{w_n}P_\theta\ket{w_n}
		}{2}
		\\
		&=
		\frac12+\frac{1}{2\sqrt{2}}
		=
		\cos^2\of[\Big]{\frac{\pi}{8}}\,.
	\end{align*}
	Together with \cref{eq:BB84-upper-bound-all-n}, this proves the result.
\end{proof}

\subsection{Strong Security Against Bounded Adversaries}
\label{subsec:Strong_Security_Against_Bounded_Adversaries}

\begin{proposition}
	\label{lem:upper-bound-in-terms-of-Bobs-dimension}
	Denote by $d\in\NN^*$ a common upper bound on the dimensions of Bob's and Charlie's Hilbert spaces.
	Then, for any $n\in\NN^*$ and any $\PauliSet^{(n)}\subseteq\set{\Identity_2,X,Y,Z}^{\otimes n}$, we have:
	$$
		\Prob^*\of[\big]{ \text{\normalfont win MoE} }
		\,\leq\,
		\frac12 + \frac{d\,\sqrt{2^n}}{2\sqrt{\abs{\PauliSet^{(n)}}}}\,.
	$$
\end{proposition}

\begin{proof}
	We use the upper bound from \cref{eq:upper-bound-on-pwin-with-the-1/2} on the winning probability:
	$$
		\Prob^*\of[\big]{ \text{\normalfont win MoE} }
		\,\leq\,
		\frac12 + \frac{1}{4\abs{\PauliSet^{(n)}}}\,
		\lambda_{\max}\of[\Bigg]{
			\underbrace{
			\sum_{P\in\PauliSet^{(n)}}
			P \otimes \of[\big]{
			V_P \otimes \Identity_\Charlie
			+ \Identity_{\Bob} \otimes V_P}
			}_{=:\,M}
		}\,.
	$$
	The maximal eigenvalue is upper-bounded by the Frobenius norm $\lambda_{\max}(M)\leq\norm{M}_2\coloneqq\sqrt{\Tr\ofAlt{MM^*}}$.
	Therefore,
	\begin{align*}
		\lambda_{\max}(M)^2
		&\,\leq\,
		\sum_{P,Q\in\PauliSet^{(n)}}
		\underbrace{\Tr\ofAlt[\big]{PQ}}_{=\,2^n\delta_{P,Q}}
		\Tr\ofAlt[\Big]{
			\of[\big]{V_P \otimes \Identity_\Charlie + \Identity_\Bob \otimes V_P}
			\of[\big]{V_Q \otimes \Identity_\Charlie + \Identity_\Bob \otimes V_Q}
		}
		\\
		&\,=\,
		2^n
		\sum_{P\in\PauliSet^{(n)}}
		\Tr\ofAlt[\Big]{
			\of[\big]{V_P \otimes \Identity_\Charlie + \Identity_\Bob \otimes V_P}^2
		}
		\\
		&\,\leq\,
		2^n\abs{\PauliSet^{(n)}}\,4d^2\,.
	\end{align*}
	In the last inequality, we used that the squared operator has operator norm at most~$4$ and acts on a space of dimension at most~$d^2$.
	The result follows by taking the square root and substituting this estimate into the upper bound above.
\end{proof}

\begin{corollary}
	\label{prop:Security-against-bounded-adversaries}
	The Pauli Encoding $\Kall^{(n)}=\set{\Identity_2,X,Y,Z}^{\otimes n}\backslash\set{\Identity_{2^n}}$ is strongly secure against families of adversaries whose local dimensions are bounded independently of~$n$, and weakly secure against families of adversaries whose local dimensions satisfy $d_n=o(\sqrt{2^n})$.
\end{corollary}

\begin{proof}
	Using \Cref{prop:Pwin-of-P_n-union-Q_n}, the winning probability for the MoE game with $\Kall^{(n)}$ is upper-bounded by the one with $\Kall^{(n)}\sqcup\set{\Identity_{2^n}}$.
	Now, using \Cref{lem:upper-bound-in-terms-of-Bobs-dimension} with $\Kall^{(n)}\sqcup\set{\Identity_{2^n}}$, we have:
	$$
		\Prob^*\of[\big]{ \text{\normalfont win MoE} }
		\,\leq\,
		\frac12 + \frac{d_n}{2\sqrt{2^n}}\,.
	$$
	Since the security parameter $\lambda=\lceil\log_2(4^n-1)\rceil$ grows linearly with~$n$, the advantage is negligible when $d_n$ is bounded independently of~$n$, and it vanishes whenever $d_n=o(\sqrt{2^n})$.
\end{proof}

\begin{remark}[Schatten Norms Analysis]
	In the above proof, we used the Frobenius norm $\norm{\,\cdot\,}_2$, 
	which is the $p=2$ particular case of the more general Schatten norms:
	$$
		\norm{M}_p\coloneqq\of[\Big]{\Tr\ofAlt[\big]{\abs{M}^p}}^{1/p}\,,
	$$ 
	where $p\in[1,\infty)$.
	This family of norms is known to converge to the spectral norm $\norm{M}_p\to\norm{M}_\infty$ as $p\to\infty$. 
	Therefore, computing the Schatten $p$-norm for large~$p$ yields a good approximation of $\norm{\,\cdot\,}_\infty$,
	and one might be tempted to compute it to improve the upper bound from \Cref{lem:upper-bound-in-terms-of-Bobs-dimension}, so that the expression no longer relies on the adversaries' dimension~$d$. 
	Nevertheless, following the discussion from \Cref{rem:Curse-of-the-3/4}, we know that there exists a strategy with local dimension $d=2^n+1$ achieving the winning probability at least $3/4$, so this method cannot work.
\end{remark}

\begin{remark}
	A similar computation with the Frobenius norm $\norm{\,\cdot\,}_2$ as above can be applied directly to the expression \cref{eq:expression-of-Pwin} 
	instead of \cref{eq:upper-bound-on-pwin-with-the-1/2}. 
	Using naively that $\abs[\big]{\Tr\ofAlt{V_P}} \leq d$ and that $\abs[\big]{\Tr\ofAlt{V_PV_Q}} \leq d$, this method leads to the following expression:
	$$
		\Prob^*\of[\big]{ \text{\normalfont win MoE} }
		\,\leq\,
		\frac14 + \frac{d\,\sqrt{2^n}}{4\sqrt{\abs{\PauliSet^{(n)}}}}\,\sqrt{8+\abs{\PauliSet^{(n)}}}\,,
	$$
	which is not conclusive.
\end{remark}

\subsection{Strong Indistinguishable Security}

In this paper, we are mainly interested in the unclonable-indistinguishable security defined in \Cref{def:Security}. 
But there exists a standard weaker notion of security called the \emph{indistinguishable security}, originally introduced in the classical framework by \citeauthor{Goldwasser-Micali-84} in \cite{Goldwasser-Micali-84} and later generalized to the quantum setting by \citeauthor{Alagic-Broadbent-Fefferman-Gagliardoni-Schaffner-StJules-16} in \cite{Alagic-Broadbent-Fefferman-Gagliardoni-Schaffner-StJules-16}, that we define and study below.
In this subsection, we show that several Pauli Encodings are strongly indistinguishable secure.

\begin{definition}[Distinguishing Game]
	In the \emph{distinguishing game}, Alice samples a classical message $x\in\set{0,1}$ and a classical key $k\in \Kset$ uniformly at random.
	She encrypts~$x$ into a quantum state $\rho_{x|k}\in\Density(\Hilbert_\Alice)$.
	She then sends the state to a pirate $\Pirate$ who applies a quantum channel
	$$
		\Phi:\Bounded(\Hilbert_\Alice)\to\Bounded(\Hilbert_\Bob)
	$$ 
	without knowing the key $k$.
	Then Bob performs a POVM measurement $\set{B_{x_\Bob}}_{x_\Bob}$ without knowing the key again.
	We say the adversary team $(\Pirate,\Bob)$ wins the distinguishing game with the strategy $(\Phi,\set{B_{x_\Bob}}_{x_\Bob})$ if Bob correctly guesses the message, \ie if
	$
		x_\Bob = x
	$.
\end{definition}

\noindent
The optimal winning probability is expressed as follows:
\begin{align*}
	\Prob^*\of[\Big]{\text{win distinguishing} \,\Big| \set{\rho_{x|k}}}
	&
	\,\coloneqq\,
	\sup_{\Phi, \set{B}}
	\frac{1}{2|\Kset|} 
	\sum_{x,x_\Bob\in\set{0,1}} 
	\sum_{k\in \Kset} 
	\Tr\ofAlt[\Big]{ \Phi\of[\big]{\rho_{x|k}} B_{x_\Bob} } \Indic_{x_\Bob = x}
	\\
	&\,=\,
	\sup_{\Phi, \set{B}}
	\frac{1}{2|\Kset|} 
	\sum_{x\in\set{0,1}} 
	\sum_{k\in \Kset} 
	\Tr\ofAlt[\Big]{ \Phi\of[\big]{\rho_{x|k}} B_x }\,.
\end{align*}

\begin{definition}[Indistinguishable Security]
	\label{def:Indistinguishable_Security}
	A family of encryption schemes $\set[\big]{\rho_{x|k}^{(\lambda)}}_{\lambda\in\NN}$ is said to be \emph{weakly indistinguishable secure} if there exists a function $f:\NN\to[0,1]$ with the vanishing condition $f(\lambda)\to0$ as $\lambda\to\infty$ such that:
	$$
		\forall \lambda\in\NN,\qquad
		\Prob^*\of[\Big]{\text{win distinguishing} \,\Big| \set[\big]{\rho_{x|k}^{(\lambda)}}}
		\,\leq\,
		\frac{1}{2} + f(\lambda)\,.
	$$
	If in addition $f$ is negligible, meaning that for every polynomial $p:\NN\to\mathbb R_{>0}$ there exists $\lambda_0\in\NN$ such that $f(\lambda)\leq 1/p(\lambda)$ for all $\lambda\geq\lambda_0$, then we say that the encryption scheme is \emph{strongly indistinguishable secure}.
\end{definition}

\noindent
Recall that Pauli Encodings are defined with keys $k$ being Pauli strings $P\in\set{\Identity_2,X,Y,Z}^{\otimes n}\backslash\set{\Identity_{2^n}}$ (see \Cref{def:Pauli_Encoding_state_version}).
We adapt the following lemma from \cite[Section~3.5]{Botteron-Broadbent-Culf-Nechita-Pellegrini-Rochette-24}:

\begin{lemma}[Universal Distinguishing Upper Bound]
	\label{lem:Universal_Distinguishing_Upper_Bound}
	For any $n\in\NN^*$ and any $\PauliSet^{(n)}\subseteq\set{\Identity_2,X,Y,Z}^{\otimes n}\backslash\set{\Identity_{2^n}}$,
	we have:
	$$
		\Prob^*\of[\Big]{\text{win distinguishing} \,\Big|\, \PauliSet^{(n)} }
		\,\leq\,
		\frac{1}{2} + \frac{1}{2K}\operatornorm[\Bigg]{\sum_{P\in\PauliSet^{(n)}}P\,}\,.
	$$
\end{lemma}

\begin{proof}
	Set $K\coloneqq\abs{\PauliSet^{(n)}}$.
	We proceed similarly as in \Cref{sec:Preliminaries}.
	We can rephrase the winning probability using the Choi matrix $C_\Phi$ of the quantum channel $\Phi$:
	\begin{align*}
		\Prob^*\of[\Big]{\text{win distinguishing} \,\Big|\, \PauliSet^{(n)}}
		&\,\leq\,
		\sup_{\substack{C_{\Phi}\succcurlyeq\zero\\\Tr\ofAlt{C_\Phi}=d_\Alice}}
		\sup_{\set{B}}
		\frac{1}{2K} 
		\sum_{x\in\set{0,1}} 
		\sum_{P\in\PauliSet^{(n)}} 
		\Tr\ofAlt[\Big]{ C_\Phi \of[\big]{\rho_{x|P}^\top \otimes B_{x} } }\,,
		\\
		&\,=\,
		\sup_{\sigma,\set{B}}
		\frac{1}{2K} 
		\sum_{x\in\set{0,1}} 
		\sum_{P\in\PauliSet^{(n)}} 
		\Tr\ofAlt[\Big]{ d_\Alice\sigma_{\Alice\Bob}\, \of[\big]{\rho_{x|P}^\top \otimes B_{x} } }\,,
		\\
		&\,=\,
		\sup_{\set{B}}
		\frac{1}{K} 
		\lambda_{\max}\of[\Bigg]{
			\sum_{x\in\set{0,1}} 
			\sum_{P\in\PauliSet^{(n)}} 
			\widetilde A_{x|P} \otimes B_{x} 
		}
		\,,
	\end{align*}
	where
	$
	\widetilde A_{x|P}
	\coloneqq
	d_\Alice\rho_{x|P}^\top/2
	=
	\of{\Identity_{2^n}+(-1)^xP^\top}/2
	$.
	Then, assuming without loss of generality that $B_x=\of{\Identity_{d_\Bob}+(-1)^xU}/2$ for some Hermitian unitary $U\in\Unitary(\Complex^{d_\Bob})$ (see \Cref{lem:change_B_and_C_into_unitaries}), we have:
	\begin{align*}
		\Prob^*\of[\Big]{\text{win distinguishing} \,\Big|\, \PauliSet^{(n)}}
		&\,\leq\,
		\sup_{U}
		\frac{1}{K} 
		\lambda_{\max}\of[\Bigg]{
			\sum_{x\in\set{0,1}} 
			\sum_{P\in\PauliSet^{(n)}} 
			\frac{1}{4}
			\of[\Big]{
				\Identity_{\Alice\Bob}
				+(-1)^x\Identity_\Alice\otimes U
				+(-1)^xP^\top\otimes\Identity_\Bob
				+P^\top\otimes U
			}
		}\,,
		\\
		&\,=\,
		\frac{1}{2}
		+
		\frac{1}{2K}
		\sup_U
		\lambda_{\max}\of[\Bigg]{
			\of[\bigg]{\sum_{P\in\PauliSet^{(n)}}P^\top}\otimes U
		}\,.
	\end{align*}
	We upper-bound the maximal eigenvalue by the operator norm and use multiplicativity of the operator norm.
	Finally, $\operatornorm{U}=1$ and $\operatornorm{\sum_P P^\top}=\operatornorm{\sum_P P}$, which proves the claimed inequality.
\end{proof}

\begin{proposition}
	\label{prop:Strong_Indistinguishable_Security}
	The following Pauli Encodings are strongly indistinguishable secure:
	\begin{enumerate}
		\item $\Kall^{(n)}=\set{\Identity_2,X,Y,Z}^{\otimes n}\backslash\set{\Identity_{2^n}}$, with
		$$
			\Prob^*\of[\Big]{\text{win distinguishing}\,\Big|\,\Kall^{(n)}}
			\,\leq\,
			\frac12+
			\frac{(1+\sqrt3)^n-1}{2(4^n-1)}
			\,\sim\,
			\frac12+\frac12\of[\bigg]{\frac{1+\sqrt3}{4}}^n.
		$$
		\item $\Kreal^{(n)}=\set{\Identity_2,X,Z}^{\otimes n}\backslash\set{\Identity_{2^n}}$, with
		$$
			\Prob^*\of[\Big]{\text{win distinguishing}\,\Big|\,\Kreal^{(n)}}
			\,\leq\,
			\frac12+
			\frac{(1+\sqrt2)^n-1}{2(3^n-1)}
			\,\sim\,
			\frac12+\frac12\of[\bigg]{\frac{1+\sqrt2}{3}}^n.
		$$
		\item $\KHZH^{(n,\Lall(n))}$, with
		$$
			\Prob^*\of[\Big]{\text{win distinguishing}\,\Big|\,\KHZH^{(n,\Lall(n))}}
			\,\leq\,
			\frac12+
			\frac{(2+\sqrt2)^n-2^n}{2(4^n-2^n)}
			\,\sim\,
			\frac12+\frac12\of[\bigg]{\frac12+\frac{1}{2\sqrt2}}^n.
		$$
		\item $\KBBeightyfour^{(n)}=\set{X,Z}^{\otimes n}$, with
		$$
			\Prob^*\of[\Big]{\text{win distinguishing}\,\Big|\,\KBBeightyfour^{(n)}}
			\,\leq\,
			\frac12+\frac{1}{2(\sqrt2)^n}.
		$$
		\item An anticommuting encoding with $K$ keys, with
		$$
			\Prob^*\of[\Big]{\text{win distinguishing}\,\Big|\,\Kanticomm}
			\,\leq\,
			\frac12+\frac{1}{2\sqrt K}.
		$$
	\end{enumerate}
\end{proposition}

\begin{proof}
	Apply \Cref{lem:Universal_Distinguishing_Upper_Bound} and compute the operator norm of the sum of the Pauli labels.
	For the first two families, these sums are respectively
	$$
		\of{\Identity_2+X+Y+Z}^{\otimes n}-\Identity_{2^n}
		\quad\text{and}\quad
		\of{\Identity_2+X+Z}^{\otimes n}-\Identity_{2^n},
	$$
	whose operator norms are $(1+\sqrt3)^n-1$ and $(1+\sqrt2)^n-1$.
	For $\KHZH^{(n,\Lall(n))}$, the sum is
	$$
		\of{2\Identity_2+X+Z}^{\otimes n}-2^n\Identity_{2^n},
	$$
	whose operator norm is $(2+\sqrt2)^n-2^n$.
	For $\KBBeightyfour^{(n)}$, the sum is $(X+Z)^{\otimes n}$ and has operator norm $(\sqrt2)^n$.
	Finally, for $K$ pairwise anticommuting Hermitian Pauli strings, the square of their sum is $K\Identity$, so its operator norm is $\sqrt K$.
	Since the security parameter is the binary key length, each displayed advantage is negligible in that parameter.
\end{proof}

\begin{remark}
	Note that, although being indistinguishable secure, 
	the Pauli Encoding with $\PauliSet^{(n)}=\set{X,Z}^{\otimes n}$
	is not unclonable-indistinguishable secure, see the table at the end of \Cref{subsec:Improved-Lower-Bounds}.
	Hence, this is an example of a separation between those two security notions.
\end{remark}

\subsection{Security Analysis of an Inefficient Protocol: the Anticommuting Encoding}
\label{subsec:Security_Analysis_of_the_Anticommuting_Protocol}

In this subsection, we study the Pauli Encoding associated with a set
$\Kanticomm=\{\Gamma_1,\ldots,\Gamma_K\}$ (\Cref{ex:Anticommuting_Encoding}) of pairwise anticommuting Hermitian Pauli strings,
\begin{equation}
    \Gamma_i^2=\Identity,
    \qquad
    \Gamma_i\Gamma_j=-\Gamma_j\Gamma_i
    \quad(i\neq j)\,.
    \label{eq:anticommuting-Pauli-relations-security}
\end{equation}
This protocol is inefficient in the security parameter $\lambda=\lceil\log_2 K\rceil$, since a representation of $K$ pairwise anticommuting Hermitian involutions requires a Hilbert-space dimension growing exponentially in~$K$.
Its conjectured winning probability is
\begin{equation}
    \Prob^*\of[\big]{\text{\normalfont win MoE}}
    =\frac12+\frac{1}{2\sqrt K}\,.
    \label{eq:anticommuting-security-conjecture}
\end{equation}
The equality is known for $K\leq 7$, and was verified numerically in \cite{Botteron-Broadbent-Culf-Nechita-Pellegrini-Rochette-24} up to $K=17$ by the second level of the NPA hierarchy.
Here, we explain a symmetry reduction that applies uniformly in $K$, obtain the limiting relaxations up to level three, and clarify what the first level can and cannot distinguish about more general Pauli protocols.

\subsubsection{A Light Formulation of the NPA Hierarchy}
\label{subsubsec:anticommuting-NPA-background}

The following formulation is adapted to the operator polynomial isolated in \cite[Proposition~5]{Botteron-Broadbent-Culf-Nechita-Pellegrini-Rochette-24}.
It is a standard noncommutative moment relaxation in the sense of \cite{NPA2008,PNA2010}.

\begin{definition}[Anticommuting Scenario Algebra]
    \label{def:anticommuting-scenario-algebra}
    For $K\in\NN^*$, let $\mathcal A_{\mathrm{ac}}(K)$ be the unital $*$-algebra generated by Hermitian elements
    $b_1,c_1,\ldots,b_K,c_K$ satisfying
    \begin{equation}
        b_i^2=c_i^2=\Identity,
        \qquad
        b_ic_i=c_ib_i,
        \qquad
        b_ic_j=-c_jb_i\quad(i\neq j).
        \label{eq:anticommuting-scenario-relations}
    \end{equation}
    We define the game polynomial
    $$
        H_K\coloneqq\sum_{i=1}^K\bigl(b_i+c_i+b_ic_i\bigr)\in\mathcal A_{\mathrm{ac}}(K).
    $$
\end{definition}

The reduction of \cite[Proposition~5]{Botteron-Broadbent-Culf-Nechita-Pellegrini-Rochette-24} shows that the optimal winning probability can be written as
\begin{equation}
    \Prob^*\of[\big]{\text{\normalfont win MoE}}
    =\frac14+\frac{1}{4K}\sup_{\pi}\lambda_{\max}\of[\big]{\pi(H_K)},
    \label{eq:winning-probability-universal-norm-HK}
\end{equation}
where the supremum is over finite-dimensional $*$-representations $\pi$ of $\mathcal A_{\mathrm{ac}}(K)$.
Thus, the security question is a largest-eigenvalue maximisation problem with polynomial relations.

\begin{definition}[Level-$t$ Moment Relaxation]
    \label{def:level-t-anticommuting-relaxation}
    Let $t\in\NN^*$.
    Denote by $\mathcal B_t(K)$ the reduced words in the generators $b_i,c_i$ of degree at most $t$, including the empty word $\Identity$.
    A level-$t$ moment functional is a linear map $L$ on reduced words of degree at most $2t$ such that
    $$
        L(\Identity)=1,
        \qquad
        L(w^*)=\overline{L(w)},
        \qquad
        M_t(L)\coloneqq\bigl(L(u^*v)\bigr)_{u,v\in\mathcal B_t(K)}\succcurlyeq\zero,
    $$
    and such that all identities induced by \cref{eq:anticommuting-scenario-relations} are imposed whenever they have degree at most $2t$.
    We set
    $$
        \beta_t(K)\coloneqq\sup_L L(H_K),
        \qquad
        \omega_t(K)\coloneqq\frac14+\frac{\beta_t(K)}{4K}.
    $$
\end{definition}

\begin{proposition}[Upper Bounds and Monotonicity]
    \label{prop:NPA-upper-bound-monotonicity-anticommuting}
    For every $K,t\in\NN^*$,
    $$
        \Prob^*\of[\big]{\text{\normalfont win MoE}}
        \leq \omega_{t+1}(K)\leq\omega_t(K).
    $$
    Moreover, the bounded noncommutative moment hierarchy converges, as $t\to\infty$, to the corresponding universal $C^*$-algebraic optimum over arbitrary Hilbert-space representations.
    This limiting value upper-bounds the finite-dimensional supremum in \cref{eq:winning-probability-universal-norm-HK}; equality with the finite-dimensional value is not needed for any finite-level upper bound used below.
\end{proposition}

\begin{proof}
    Every concrete representation $\pi$ and unit vector $\ket{\psi}$ define the moments
    $L(w)=\bra{\psi}\pi(w)\ket{\psi}$, whose moment matrices are Gram matrices and are therefore positive semidefinite.
    Hence every physical strategy is feasible at every finite level, which proves the first inequality.
    Increasing $t$ adds moments and consistency constraints, so the feasible projections form a decreasing sequence, proving the second inequality.
    Convergence to the universal operator-algebraic optimum follows from the standard convergence theorem for bounded noncommutative polynomial optimisation: the involution relations in \cref{eq:anticommuting-scenario-relations} make all generators uniformly bounded; see \cite{PNA2010} and, for the correlation formulation from which the hierarchy originates, \cite{NPA2008}.
\end{proof}

The point of the hierarchy is therefore not that a finite level already describes all operator strategies.
Rather, each level gives a computable upper bound, and higher levels incorporate longer products of the observables.

\subsubsection{The Anticommuting Protocol is Optimal at Level One}
\label{subsubsec:anticommuting-optimal-level-one}

At level one, the moment matrix is indexed only by $\Identity,b_1,\ldots,b_K,c_1,\ldots,c_K$.
This is exactly the Gram-matrix formulation in \cite[Eq.~(19)]{Botteron-Broadbent-Culf-Nechita-Pellegrini-Rochette-24}.
After averaging over permutations of the keys and over the exchange $b_i\leftrightarrow c_i$, it reduces to the three-variable problem of \cite[Lemma~6 and Eq.~(20)]{Botteron-Broadbent-Culf-Nechita-Pellegrini-Rochette-24}.

\begin{proposition}[Exact Level-One Value for the Anticommuting Protocol]
    \label{prop:exact-level-one-anticommuting-value}
    The level-one upper bound is
    $$
        \omega_1(K)=
        \begin{cases}
            \displaystyle \frac12+\frac{1}{2\sqrt K},&2\leq K\leq7,\\[2ex]
            \displaystyle \frac58+\frac{1}{2(K-2)}-\frac{1}{4K},&K\geq8.
        \end{cases}
    $$
    In particular,
    \begin{equation}
        \lim_{K\to\infty}\omega_1(K)=\frac58.
        \label{eq:level-one-asymptotic-five-eighths}
    \end{equation}
\end{proposition}

\begin{proof}
    This is \cite[Theorem~9]{Botteron-Broadbent-Culf-Nechita-Pellegrini-Rochette-24}.
\end{proof}

The significance of \cref{eq:level-one-asymptotic-five-eighths} extends beyond the anticommuting protocol.
Let $P_1,\ldots,P_K$ be arbitrary Hermitian Pauli strings and write
\begin{equation}
    P_iP_j=\varepsilon_{ij}P_jP_i,
    \qquad
    \varepsilon_{ij}\in\{+1,-1\},
    \qquad
    \varepsilon_{ii}=1.
    \label{eq:Pauli-commutation-sign-matrix}
\end{equation}
The corresponding transformed scenario algebra satisfies $b_ic_j=\varepsilon_{ij}c_jb_i$.
The level-one relaxation sees this relation only through the entries between the vectors $b_i\ket{\psi}$ and $c_j\ket{\psi}$.

\begin{proposition}[Best Level-One Bound Among Pauli Commutation Patterns]
    \label{prop:anticommuting-best-at-level-one}
    Let $\omega_1^{\varepsilon}(K)$ be the level-one upper bound associated with the commutation-sign matrix $\varepsilon=(\varepsilon_{ij})$ in \cref{eq:Pauli-commutation-sign-matrix}.
    Then
    \begin{equation}
        \omega_1^{\varepsilon}(K)\geq\omega_1^{\mathrm{ac}}(K),
        \label{eq:all-anticommuting-minimises-level-one}
    \end{equation}
    where $\omega_1^{\mathrm{ac}}(K)$ denotes the value for which $\varepsilon_{ij}=-1$ whenever $i\neq j$.
    Consequently, no level-one analysis of a Pauli protocol can produce an asymptotic upper bound below $5/8$.
\end{proposition}

\begin{proof}
    Let $G$ be a feasible complex Hermitian level-one moment matrix.
    For $i,j\in[K]$, define $z_{ij}\coloneqq G_{b_i,c_j}=L(b_ic_j)$.
    Since $(b_ic_j)^*=c_jb_i=\varepsilon_{ij}b_ic_j$, one has
    \begin{equation}
        \overline{z_{ij}}=\varepsilon_{ij}z_{ij}.
        \label{eq:real-imaginary-level-one-entry}
    \end{equation}
    Thus $z_{ij}$ is real when $P_i$ and $P_j$ commute and purely imaginary when they anticommute.

    The objective and all diagonal constraints are real.
    We may therefore replace $G$ by its entrywise real part $\Re(G)$.
    This matrix remains positive semidefinite: for every $x,y\in\mathbb R^{1+2K}$,
    \begin{equation}
        (x+\mathrm i y)^*\Re(G)(x+\mathrm i y)
        =x^{\mathsf T}\Re(G)x+y^{\mathsf T}\Re(G)y\geq0,
    \end{equation}
    because $x^{\mathsf T}\Re(G)x=x^*Gx\geq0$, and similarly for $y$.
    After this realification, \cref{eq:real-imaginary-level-one-entry} gives that $z_{ij}=0$ when $\varepsilon_{ij}=-1$, whereas a commuting pair leaves $z_{ij}$ as a free real moment.

    Hence the all-anticommuting pattern imposes all the zero constraints imposed by any other Pauli commutation pattern, and possibly more.
    Every real level-one moment matrix feasible for the all-anticommuting pattern is therefore feasible for the pattern $\varepsilon$.
    The feasible-set inclusion gives \cref{eq:all-anticommuting-minimises-level-one}.
    The final assertion follows from \cref{eq:level-one-asymptotic-five-eighths}.
\end{proof}

\begin{remark}
    \label{rem:level-one-blindness-other-Pauli-constraints}
    This proposition does not say that the anticommuting protocol is the most secure Pauli protocol.
    It says that it is the most secure protocol that the first NPA level can distinguish (using only degree-two moments).
    Relations involving products such as $b_ib_j$, $c_ic_j$, or $b_ic_jb_kc_\ell$ first become visible at higher levels.
    Thus, the $5/8$ obstruction is a limitation of the relaxation, not a cryptographic impossibility result.
\end{remark}

\subsubsection{Symmetry-Reduced Asymptotic Bounds Through Level Three}
\label{subsubsec:anticommuting-asymptotic-level-three}

The level-two construction of \cite[Section~4.3]{Botteron-Broadbent-Culf-Nechita-Pellegrini-Rochette-24} exploits symmetry by identifying moment entries that belong to the same signed orbit.
This reduces the number of scalar variables from $1+3K^2$ to the eighteen parameters $g_1,\ldots,g_{18}$ displayed there.
That orbit reduction is extremely convenient, and we retain exactly the same moment convention.
However, it leaves one positive-semidefinite constraint on a Gram matrix whose dimension still grows with $K$.

Our complementary reduction applies the Wedderburn decomposition to the invariant Gram matrix itself.
The relevant symmetry group is
$$
	\mathfrak S_K\times C_2,
$$
where $\mathfrak S_K$ simultaneously permutes the indices and the non-trivial element of $C_2$ exchanges the $b$- and $c$-families, including the signs required to restore normal order.
Averaging a feasible moment functional over this finite group preserves both feasibility and the objective.
The moment matrix therefore belongs to the commutant of its action, and the standard symmetry reduction for semidefinite programs applies \cite{GatermannParrilo2004,Serre1977}.

\begin{proposition}[Multiplicity-Block Reduction]
    \label{prop:anticommuting-multiplicity-block-reduction}
    Fix a hierarchy level $t$.
    In the stable range of $K$, the reduced word space admits an orthogonal decomposition
    \begin{equation}
        \mathbb R^{\mathcal B_t(K)}
        \cong
        \bigoplus_{\lambda,\epsilon}
        S^\lambda\otimes\mathbb R^{m^{(t)}_{\lambda,\epsilon}},
        \label{eq:anticommuting-isotypic-decomposition}
    \end{equation}
    where $\lambda\vdash K$, $S^\lambda$ is the corresponding real Specht module, and $\epsilon\in\{+1,-1\}$ is the eigenvalue of the family exchange.
    Every invariant moment matrix has the form
    $$
        M_t(g)
        \cong
        \bigoplus_{\lambda,\epsilon}
        \Identity_{S^\lambda}\otimes B^{(K)}_{\lambda,\epsilon}(g).
    $$
    Consequently,
    \begin{equation}
        M_t(g)\succcurlyeq\zero
        \quad\Longleftrightarrow\quad
        B^{(K)}_{\lambda,\epsilon}(g)\succcurlyeq\zero
        \quad\text{for every occurring }(\lambda,\epsilon).
        \label{eq:anticommuting-block-PSD-equivalence}
    \end{equation}
    At fixed $t$, the multiplicity dimensions $m^{(t)}_{\lambda,\epsilon}$ stabilize as $K\to\infty$.
\end{proposition}

\begin{proof}
    The statement is the real commutant decomposition and Schur's lemma.
    Representation stability follows because a level-$t$ word involves at most $t$ indices: only partitions of the form $\lambda=(K-d,\nu)$ with $0\leq d\leq t$ can occur.
    A detailed construction of the blocks, including the exact multiplicity formula and the stable choice of coordinates, is given in \Cref{app:anticommuting-NPA-reduction}.
    See also \cite{ChurchEllenbergFarb2015,Sagan2001,Macdonald1995} for the underlying stable representation theory.
\end{proof}

The distinction with the orbit method is therefore as follows.
Orbit reduction identifies equal scalar entries of an invariant matrix; Wedderburn reduction decomposes the invariant matrix into its irreducible positive-semidefinite blocks.
The two procedures are compatible and complementary: the former provides a concise set of variables, while the latter removes the growth of the semidefinite block dimensions at fixed hierarchy level.

Our implementation can be found in the accompanying GitHub repository \cite{GitHub}.
It first reduces words and signed moment orbits exactly.
It computes the Specht multiplicities from irreducible characters, evaluated by the Murnaghan--Nakayama rule, and constructs the multiplicity spaces with Young--Jucys--Murphy projectors \cite{Jucys1974,Murphy1981,OkounkovVershik1996}.
Importantly, it selects one family of stable projected seed vectors simultaneously for all sampled values of $K$ and does not orthonormalize it separately at each $K$.
If $V^{(K)}_{\lambda,\epsilon}$ is the resulting full-rank seed matrix, the implemented block is the exact congruence
$$
    B^{(K)}_{\lambda,\epsilon}(g)
    =\bigl(V^{(K)}_{\lambda,\epsilon}\bigr)^{\mathsf T}
    M_t(g)V^{(K)}_{\lambda,\epsilon}.
$$
This avoids $K$-dependent square roots and Gram--Schmidt sign choices and gives a single coordinate convention suitable for exact interpolation.
The finite-$K$ coefficients are obtained in exact arithmetic, interpolated as rational functions of $K$, checked at independent values, and specialized either at a finite value of $K$ or at $K=\infty$ after an exact symbolic congruence.
The technical construction and a level-two example are recorded in \Cref{app:anticommuting-NPA-reduction}.

\begin{proposition}[Existence of the Stable Reduced Relaxations]
    \label{prop:stable-reduced-relaxations-level-three}
    For each $t\in\{1,2,3\}$, the finite-$K$ level-$t$ relaxation admits, in the stable range, a finite family of multiplicity-block constraints whose dimensions do not depend on $K$.
    After the exact block congruence used in the implementation, every coefficient has a finite limit as $K\to\infty$.
    Denote its optimum by $\omega_t^{(\infty)}$.
    Then
    \begin{equation}
        \limsup_{K\to\infty}
        \Prob^*\of[\big]{\text{\normalfont win MoE}}
        \leq
        \omega_t^{(\infty)}.
        \label{eq:limsup-security-bound}
    \end{equation}
\end{proposition}

\begin{proof}
    The finite block equivalence is \cref{eq:anticommuting-block-PSD-equivalence}.
    At fixed $t$, only finitely many equality patterns of at most $2t$ indices occur.
    Their orbit cardinalities are falling-factorial polynomials in $K$, while the Young--Jucys--Murphy projectors have rational coefficients in $K$ in the stable coordinate convention.
    Hence the block coefficients are rational functions of $K$.
    The symbolic congruence removes their support-dependent leading powers, so coefficientwise limits exist.
    The moment variables are uniformly bounded, because they are entries of positive semidefinite moment matrices with unit diagonal. Taking a convergent subsequence of finite-$K$ optimizers and passing coefficientwise to the limit therefore gives a feasible point of the limiting SDP, which proves the upper bound in \cref{eq:limsup-security-bound}.
    Further details are provided in \Cref{app:anticommuting-NPA-reduction}.
\end{proof}

The exact construction of the limiting SDPs was followed by a numerical solution of those problems.
The values obtained for these asymptotic upper bounds are
$$
    \begin{array}{c|l|c}
        \text{Level }t
        & \displaystyle\omega_t^{(\infty)}
        & \text{Status}\\[1ex]
        \hline
        1 & 0.625=5/8 & \text{exact}\\
        2 & 0.579982991 & \text{numerical}\\
        3 & 0.555608131 & \text{numerical}
    \end{array}
$$
The successive improvements are substantial and are consistent with the conjectured limit $1/2$ from \cref{eq:anticommuting-security-conjecture}.
They do not, by themselves, prove this conjecture: the values at levels two and three are numerical optima of rigorously specified limiting SDPs.
Note that our results for level two reproduce those of \cite{Botteron-Broadbent-Culf-Nechita-Pellegrini-Rochette-24}, namely, that the value from \cref{eq:anticommuting-security-conjecture} is numerically attained up to $K=17$.
Moreover, our symmetry reduction of level three shows numerically that this value is obtained up to $K=40$.

\subsubsection{Current Status at Level Four}
\label{subsubsec:anticommuting-level-four}

The same construction has been implemented at level four.
Mathematically, this level introduces reduced words of degree four and moments of degree eight.
New stable Specht shapes with four boxes below the first row appear, and the twisted character of the family exchange acquires contributions from words containing two $b$-letters and two $c$-letters.
Thus, level four genuinely detects identifications and signs that are invisible at levels two and three.

The finite-$K$ code is functional, but the exact construction is currently prohibitively slow.
The reduced word space grows as $K^4$, while a direct classification of its moment entries has an intrinsic $K^8$ scale before symmetry and sparsity are exploited.
The implementation uses canonical equality patterns, packed moment-entry buckets, exact character calculations, and a common stable seed convention to limit memory consumption.
Nevertheless, computations beyond $K=10$ take too long for the collection of sufficiently many exact interpolation and validation points.
We therefore do not presently report a general $K$-dependent level-four SDP or an asymptotic level-four value.

This computational obstruction is not evidence of a failure of the method.
It reflects the cost of producing exact, convention-compatible blocks rather than a single floating-point relaxation at a fixed value of $K$.
The level-four implementation already supplies finite-$K$ relaxations and a precise framework for a future asymptotic computation, but further combinatorial compression of the coefficient projection is required before that computation becomes practical.

\subsection{Security Analysis of Some Efficient Protocols}
\label{subsec:Security_Analysis_of_some_Efficient_Protocols}

\subsubsection{Numerical Evidence Towards Strong Security}
\label{subsubsec:numerical-evidence-efficient-protocols}

We now apply the first level of the NPA hierarchy described in
\Cref{subsec:Security_Analysis_of_the_Anticommuting_Protocol}
to the efficient Pauli Encodings introduced in
\Cref{subsec:Examples-of-Pauli-Encodings}.
The relaxation only uses the commutation or anticommutation relation between every pair of Pauli strings.
The Julia implementation can be found in the accompanying repository \cite{GitHub}.

For computational convenience, the identity is adjoined whenever this restores a simple tensor-product description.
If the augmented labelled family contains $K$ operators, of which $r$ are identity labels, then \Cref{prop:Pwin-of-P_n-union-Q_n} gives directly
\begin{equation}
    \Prob^*\of[\big]{\text{\normalfont win MoE}\mid\PauliSet}
    \;\leq\;
    \frac{K\,p_{\mathrm{NPA},1}-r}{K-r},
    \label{eq:efficient-protocol-affine-identity-removal}
\end{equation}
where $p_{\mathrm{NPA},1}$ is the level-one value computed for the augmented family.
Here $r=1$ for $\set{\Identity_2,X,Z}^{\otimes n}$ and $\set{\Identity_2,X,Y,Z}^{\otimes n}$, whereas $r=2^n$ for the augmented $HZH$ family obtained by adjoining $0^n$ to $\Lall(n)$; no correction is required for $\Lsym(w,n)$.
The table below reports the affine-corrected outputs of the SDP after removing all identity labels.

\subsubsection[Some Values at Finite K]{Some Values at Finite $K$}
\label{subsubsec:finite-K-efficient-protocol-values}

\begin{table}[t]
    \centering
    \small
    \setlength{\tabcolsep}{5pt}
    \renewcommand{\arraystretch}{1.45}
    \begin{tabular}{lccccc}
        \shortstack{Pauli\\encoding} & $K$ & \shortstack{Identity\\attack} && \shortstack{NPA\\level one} & \shortstack{Conjectured\\exact value} \\ \hline

        \rowcolor{gray1} \multicolumn{6}{c}{$n=1$} \\
        $\Kreal^{(1)}$ & $2$ & $0.8536$ & $\approx$ & $0.8536$ & $(2+\sqrt{2})/4$ \\
        $\KHZH^{(1,\Lsym(1,1))}$ & $2$ & $0.8536$ & $\approx$ & $0.8536$ & $(2+\sqrt{2})/4$ \\
        $\KHZH^{(1,\Lall(1))}$ & $2$ & $0.8536$ & $\approx$ & $0.8536$ & $(2+\sqrt{2})/4$ \\
        $\Kall^{(1)}$ & $3$ & $0.7887$ & $\approx$ & $0.7887$ & $(3+\sqrt{3})/6$ \\

        \hline
        \rowcolor{gray1} \multicolumn{6}{c}{$n=2$} \\
        $\Kreal^{(2)}$ & $8$ & $0.8018$ & $\approx$ & $0.8018$ & $(5+\sqrt{2})/8$ \\
        $\KHZH^{(2,\Lall(2))}$ & $12$ & $0.8190$ & $\approx$ & $0.8190$ & $(7+2\sqrt{2})/12$ \\
        $\Kall^{(2)}$ & $15$ & $0.7155$ & $\neq$ & $0.7236$ & $(5+\sqrt{5})/10$ \\

        \hline
        \rowcolor{gray1} \multicolumn{6}{c}{$n=3$} \\
        $\KHZH^{(3,\Lsym(1,3))}$ & $24$ & $0.8536$ & $\approx$ & $0.8536$ & $(2+\sqrt{2})/4$ \\
        $\KHZH^{(3,\Lsym(2,3))}$ & $24$ & $0.7500$ & $\approx$ & $0.7500$ & $3/4$ \\
        $\Kreal^{(3)}$ & $26$ & $0.7514$ & $\approx$ & $0.7514$ & $(32+5\sqrt{2})/52$ \\
        $\KHZH^{(3,\Lall(3))}$ & $56$ & $0.7839$ & $\approx$ & $0.7839$ & $(34+7\sqrt{2})/56$ \\
        $\Kall^{(3)}$ & $63$ & $0.6539$ & $\neq$ & $0.6687$ & $337/504$ \\

        \hline
        \rowcolor{gray1} \multicolumn{6}{c}{$n=4$} \\
        $\KHZH^{(4,\Lsym(1,4))}$ & $64$ & $0.8536$ & $\approx$ & $0.8536$ & $(2+\sqrt{2})/4$ \\
        $\KHZH^{(4,\Lsym(3,4))}$ & $64$ & $0.6768$ & $\neq$ & $0.6771$ & $65/96$ \\
        $\Kreal^{(4)}$ & $80$ & $0.7061$ & $\approx$ & $0.7061$ & $(24+3\sqrt{2})/40$ \\
        $\KHZH^{(4,\Lsym(2,4))}$ & $96$ & $0.7500$ & $\approx$ & $0.7500$ & $3/4$ \\
        $\KHZH^{(4,\Lall(4))}$ & $240$ & $0.7498$ & $\approx$ & $0.7498$ & $(73+12\sqrt{2})/120$ \\
        $\Kall^{(4)}$ & $255$ & $0.6073$ & $\neq$ & $0.6436$ & $1313/2040$ \\
        \hline
    \end{tabular}
    \caption{
        Identity-attack probabilities (lower bounds) and numerical level-one NPA upper bounds after removing all identity labels.
        Here, $K$ denotes the resulting number of keys.
        The exact expressions in the last column are conjectures for the level-one NPA values.
        As explained in \Cref{subsubsec:anticommuting-optimal-level-one}, the level-one NPA values are lower bounded by $5/8=0.625$ for any Pauli Encoding.
    }
    \label{tab:level-one-efficient-Pauli-protocols}
\end{table}

The identity attack numerically saturates the level-one relaxation for $\set{\Identity_2,X,Z}^{\otimes n}$ for every $1\leq n\leq4$ considered here.
The same is true for the $HZH$ family associated with $\Lall(n)$, now also for $n=4$.
Since the affine map in \eqref{eq:efficient-protocol-affine-identity-removal} preserves equality, the same numerical saturation holds after removing the identity labels.
This provides strong evidence that, at level one, the behaviour of these two families is completely captured by the simplest attack.

The situation is different for $\set{\Identity_2,X,Y,Z}^{\otimes n}$: from $n=2$ onwards, the level-one value is strictly larger than the value of the identity attack.
Nevertheless, the numerical upper bounds decrease with~$n$.
At this stage, neither the conjectured exact values nor their asymptotic behaviour are proved.

Finally, the fixed-weight $HZH$ families show that increasing the number of labels does not by itself guarantee security.
For fixed~$w$, the identity attack has winning probability
\begin{equation}
    \Prob\of[\big]{\text{\normalfont win MoE}\mid V_P=\Identity}
    =
    \frac12+\frac{1}{2^{w/2+1}},
\end{equation}
independently of~$n$.
Thus, $w$ must grow with~$n$ for this family to be a candidate for strong security.

\printbibliography

\appendix
\addtocontents{toc}{\protect\setcounter{tocdepth}{1}}

\small
\section{\texorpdfstring{\!}{}Exact Symmetry Reduction for the Anticommuting NPA Hierarchy}
\label{app:anticommuting-NPA-reduction}

This appendix records the mathematical construction whose implementation accompanies this article \cite{GitHub}.
The purpose is to justify the finite-$K$ block reduction and the passage to stable large-$K$ coordinates without reproducing the lengthy coefficient matrices generated by the code.
Throughout, $K$ denotes the number of pairs of generators.

\subsection{Reduced Words and Signed Symmetries}

Let $\mathcal A_{\mathrm{ac}}(K)$ be the algebra from \cref{def:anticommuting-scenario-algebra}.
Every word can be put in the form
\begin{equation}
    w=b_{\alpha_1}\cdots b_{\alpha_p}c_{\beta_1}\cdots c_{\beta_q},
    \qquad
    \alpha_r\neq\alpha_{r+1},
    \qquad
    \beta_s\neq\beta_{s+1},
    \label{eq:appendix-normal-word}
\end{equation}
by moving every $b$-letter to the left of every $c$-letter and cancelling adjacent equal letters within each family.
A crossing of $c_i$ with $b_j$ contributes the sign $-1$ precisely when $i\neq j$.

\begin{lemma}[Uniqueness of the Normal Form]
    \label{lem:unique-normal-form-anticommuting}
    The rewriting procedure above is terminating and confluent.
    Consequently, every element represented by a word has a unique signed normal form of the type \cref{eq:appendix-normal-word}.
\end{lemma}

\begin{proof}
    Moving a $c$-letter to the right decreases the number of out-of-order pairs, while cancelling $b_i^2$ or $c_i^2$ decreases the word length.
    Hence no infinite reduction sequence exists.
    The only non-disjoint ambiguities arise when a crossing is adjacent to a cancellation.
    Resolving the cancellation first or moving the letter through first produces the same result, because the crossing sign occurs twice and therefore squares to $1$.
    All ambiguities are resolvable, so Bergman's Diamond Lemma applies \cite{Bergman1978}.
\end{proof}

For $t\in\NN$, let $\mathcal B_t(K)$ be the normal words of degree at most $t$ and let $\mathcal W_t(K):=\mathbb R^{\mathcal B_t(K)}$ with its canonical orthonormal word basis.
If $a_0(x)=1$ and $a_p(x)=x(x-1)^{p-1}$ for $p\geq1$, then
\begin{equation}\nonumber
    \dim\mathcal W_t(K)=\sum_{p+q\leq t}a_p(K)a_q(K).
\end{equation}
In particular, the dimension grows as $K^t$ at fixed $t$.

The symmetric group $\mathfrak S_K$ acts by simultaneous relabelling of the indices.
Let $\tau$ exchange the two generator families and restore normal order:
\begin{equation}\nonumber
    \tau\bigl(
        b_{\alpha_1}\cdots b_{\alpha_p}
        c_{\beta_1}\cdots c_{\beta_q}
    \bigr)
    =(-1)^{\#\{(r,s):\alpha_r\neq\beta_s\}}
        b_{\beta_1}\cdots b_{\beta_q}
        c_{\alpha_1}\cdots c_{\alpha_p}.
\end{equation}
The actions of $\mathfrak S_K$ and $\tau$ commute, giving an orthogonal action of
$G_K:=\mathfrak S_K\times C_2$.
The moment variables are indexed by signed orbits under simultaneous relabelling, $\tau$, and the adjoint $w\mapsto w^*$.
If an orbit contains the same normal word with both signs, its real moment is forced to vanish.
This recovers the eighteen-variable convention of \cite[Section~4.3]{Botteron-Broadbent-Culf-Nechita-Pellegrini-Rochette-24} at level two and extends it canonically at higher levels.

\subsection{Characters and Multiplicities}

For $\sigma\in\mathfrak S_K$, let $X_r(\sigma)$ denote the number of $r$-cycles of $\sigma$.

\begin{lemma}[Character of the Word Representation]
    \label{lem:appendix-word-character}
    The character of $\mathcal W_t(K)$ is
    \begin{equation}\nonumber
        \chi_t(\sigma)
        =\sum_{p+q\leq t}
        a_p\bigl(X_1(\sigma)\bigr)
        a_q\bigl(X_1(\sigma)\bigr).
    \end{equation}
\end{lemma}

\begin{proof}
    A normal word is fixed by $\sigma$ if and only if the index in every letter is a fixed point of $\sigma$.
    For a positive length $p$, there are $X_1(\sigma)(X_1(\sigma)-1)^{p-1}$ reduced $b$-sequences, and similarly for the $c$-sequence.
    Summing over $p+q\leq t$ proves the formula.
\end{proof}

To split the multiplicity according to the eigenvalue of $\tau$, define the twisted character
\begin{equation}\nonumber
    \theta_t(\sigma):=\operatorname{Tr}_{\mathcal W_t(K)}(\sigma\tau).
\end{equation}
Up to level four, direct counting gives
\begin{equation}\nonumber
    \theta_t(\sigma)
    =1
    +\mathbf 1_{t\geq2}\bigl(X_1(\sigma)-2X_2(\sigma)\bigr)
    +\mathbf 1_{t\geq4}
       \bigl(X_1(\sigma)(X_1(\sigma)-1)-4X_1(\sigma)X_2(\sigma)
       +4X_2(\sigma)^2-2X_2(\sigma)\bigr).
\end{equation}
Indeed, a diagonal contribution to $\sigma\tau$ requires equal $b$- and $c$-lengths.
The first non-trivial contribution comes from length $(1,1)$, and the next one from length $(2,2)$.

Let $\chi^\lambda$ be the irreducible character of the Specht module $S^\lambda$, for $\lambda\vdash K$.
The multiplicities in \cref{eq:anticommuting-isotypic-decomposition} are
\begin{equation}
    m^{(t)}_{\lambda,\pm}
    =\frac12\left(
        \left\langle\chi_t,\chi^\lambda\right\rangle_{\mathfrak S_K}
        \pm
        \left\langle\theta_t,\chi^\lambda\right\rangle_{\mathfrak S_K}
    \right).
    \label{eq:appendix-multiplicity-formula}
\end{equation}
The implementation evaluates \cref{eq:appendix-multiplicity-formula} exactly by summing over cycle types and computing $\chi^\lambda$ with the Murnaghan--Nakayama rule; see \cite[Chapter~4]{Sagan2001} and \cite[Section~I.7]{Macdonald1995}.
No table of multiplicities or hard-coded character polynomial is required.

\begin{proposition}[Stable Shapes]
    \label{prop:appendix-stable-shapes}
    At hierarchy level $t$, only partitions of the form
    \begin{equation}\nonumber
        \lambda=(K-d,\nu),
        \qquad
        0\leq d\leq t,
        \qquad
        \nu\vdash d,
        \qquad
        K-d\geq\nu_1,
    \end{equation}
    can occur.
    For fixed $(d,\nu)$ and fixed $t$, the multiplicities in \cref{eq:appendix-multiplicity-formula} are independent of $K$ for all sufficiently large~$K$.
\end{proposition}

\begin{proof}
    Each basis word has support at most $t$.
    Therefore $\mathcal W_t(K)$ is generated by permutation representations induced from subgroups fixing all but at most $t$ points.
    Young's rule implies that an irreducible constituent can have at most $t$ boxes below its first row.
    Stability of the resulting multiplicities is a standard instance of representation stability for symmetric groups; see \cite{ChurchEllenbergFarb2015}.
\end{proof}

\subsection{Stable Young--Jucys--Murphy Coordinates}

Fix a stable partition $\lambda=(K-d,\nu)$ and let $T_\lambda$ be its row-standard tableau.
For $m\in[K]$, define the Young--Jucys--Murphy element
\begin{equation}\nonumber
    J_m:=\sum_{r<m}(r\ m)\in\mathbb R[\mathfrak S_K].
\end{equation}
The elements $J_m$ commute and are self-adjoint.
On Young's seminormal basis, $J_m$ acts with eigenvalue equal to the content of the box occupied by $m$ \cite{Jucys1974,Murphy1981,OkounkovVershik1996}.

Let $T_\lambda^{(m)}$ be the restriction of $T_\lambda$ to $[m]$, let $x_m$ be the box added at step $m$, and let $\operatorname{Add}(T_\lambda^{(m-1)})$ be the addable boxes of the prefix shape.
The primitive orthogonal idempotent selecting the $T_\lambda$ seminormal line is obtained recursively by
\begin{equation}\nonumber
    E_{T_\lambda}
    =\prod_{m=1}^K
      \prod_{y\in\operatorname{Add}(T_\lambda^{(m-1)})\setminus\{x_m\}}
      \frac{J_m-\operatorname{ct}(y)}
           {\operatorname{ct}(x_m)-\operatorname{ct}(y)}.
\end{equation}
The alternative contents must be taken from the actual prefix shape; this is essential for the spectral projection to follow the chosen tableau branch.
Let
\begin{equation}\nonumber
    P_\epsilon:=\frac12(\Identity+\epsilon\tau),
    \qquad
    \epsilon\in\{+1,-1\}.
\end{equation}
Since the two group factors commute, $E_{T_\lambda}$ and $P_\epsilon$ are commuting orthogonal projectors.
Their common image has dimension $m^{(t)}_{\lambda,\epsilon}$.

The implementation uses stable word descriptors.
A positive descriptor labels one of the bounded number of indices in the long first row of $T_\lambda$, while a negative descriptor labels one of the $d$ boxes in its bounded tail.
The same descriptor can therefore be instantiated for every sufficiently large $K$.
The code scans these descriptors in a fixed support-refining order and retains a seed whenever its projection
\begin{equation}\nonumber
    v_r^{(K)}:=P_\epsilon E_{T_\lambda}e_{w_r^{(K)}}
\end{equation}
raises the exact rank.
The scan stops when the multiplicity from \cref{eq:appendix-multiplicity-formula} is reached.
A single seed convention is certified on the selected values of $K$ and reused throughout one interpolation.

\begin{proposition}[Raw Projected Multiplicity Block]
    \label{prop:appendix-raw-projected-block}
    Let
    \begin{equation}\nonumber
        V_{\lambda,\epsilon}^{(K)}
        :=\bigl[v_1^{(K)}\ \cdots\ v_m^{(K)}\bigr],
        \qquad
        m=m^{(t)}_{\lambda,\epsilon},
    \end{equation}
    and assume that this matrix has full column rank.
    For every $G_K$-invariant symmetric moment matrix $M_t(g)$, define
    \begin{equation}
        B_{\lambda,\epsilon}^{(K)}(g)
        :=\bigl(V_{\lambda,\epsilon}^{(K)}\bigr)^{\mathsf T}
          M_t(g)V_{\lambda,\epsilon}^{(K)}.
        \label{eq:appendix-raw-block}
    \end{equation}
    Then
    \begin{equation}
        M_t(g)\succcurlyeq\zero
        \quad\Longleftrightarrow\quad
        B_{\lambda,\epsilon}^{(K)}(g)\succcurlyeq\zero
        \quad\text{for every occurring }(\lambda,\epsilon).
        \label{eq:appendix-raw-block-equivalence}
    \end{equation}
\end{proposition}

\begin{proof}
    The projector $P_\epsilon E_{T_\lambda}$ selects one seminormal line in every copy of $S^\lambda\boxtimes\epsilon$.
    Its image is therefore canonically isomorphic to the multiplicity space.
    Full column rank means that the columns of $V_{\lambda,\epsilon}^{(K)}$ form a basis of this image.
    Restriction to this image gives the multiplicity operator in a non-orthonormal basis, and \cref{eq:appendix-raw-block} is its Gram-congruence matrix.
    Positivity is invariant under an invertible change of coordinates, and Schur's lemma then gives \cref{eq:appendix-raw-block-equivalence}; see \cite{GatermannParrilo2004,Serre1977}.
\end{proof}

\subsection{A Pedagogical Level-Two Block Entry}

The following calculation illustrates why the Wedderburn decomposition contains information that is not visible from a list of orbit variables alone.
Let $u=(u_1,\ldots,u_K)\in\mathbb R^K$ satisfy
\begin{equation}\nonumber
    \sum_{i=1}^K u_i=0,
    \qquad
    \sum_{i=1}^K u_i^2=1.
\end{equation}
Thus $u$ belongs to the standard representation $S^{(K-1,1)}$.
Consider the family-exchange-even unit vector
\begin{equation}\nonumber
    q_u:=\frac{1}{\sqrt2}
    \sum_{i=1}^K u_i\bigl(e_{b_i}+e_{c_i}\bigr).
\end{equation}
Use the level-two orbit convention $g_2=L(b_ic_i)$ and $g_3=L(b_ib_j)=L(c_ic_j)$ (for $i\neq j$).
The anticommutation relations and realification give $L(b_ic_j)=0$ for $i\neq j$.
Therefore
\begin{equation}\nonumber
    \left\langle q_u,M_2(g)q_u\right\rangle=\frac12\sum_{i,j}u_iu_j\bigl(L(b_ib_j)+L(c_ic_j)+L(b_ic_j)+L(c_ib_j)\bigr) =1+g_2+g_3\sum_{i\neq j}u_iu_j =1+g_2-g_3,
\end{equation}
where the last equality uses
\begin{equation}\nonumber
    \sum_{i\neq j}u_iu_j
    =\left(\sum_i u_i\right)^2-\sum_i u_i^2=-1.
\end{equation}

The term $-g_3$ is the contribution of centring in the standard representation.
An anchored vector such as $e_{b_1}+e_{c_1}$ does not see it.
This example explains the conceptual difference between the orbit-based reduction of \cite[Section~4.3]{Botteron-Broadbent-Culf-Nechita-Pellegrini-Rochette-24}, which identifies the scalar $g_3$, and the Wedderburn reduction, which determines how $g_3$ acts on each irreducible sector.
The implemented raw projected block is congruent to this transparent orthonormal-coordinate calculation.

\subsection{Exact Interpolation and the Large-\texorpdfstring{$K$}{K} Limit}

At fixed hierarchy level, each entry of \cref{eq:appendix-raw-block} is assembled from finitely many equality patterns.
The number of instantiations of a support pattern is a falling factorial in $K$.
The coefficients of the Young--Jucys--Murphy projector are rational functions of the tableau contents and hence rational functions of $K$ in the stable convention.
It follows that every block coefficient is a rational function of $K$.

The implementation evaluates the finite blocks in exact rational arithmetic at a sequence of values of $K$.
It then interpolates the non-zero entries as rational functions and checks the resulting models at independent exact values of $K$.
For the large-$K$ specialization, the constant Gram block
\begin{equation}\nonumber
    G_{\lambda,\epsilon}^{(K)}
    :=
    \bigl(V_{\lambda,\epsilon}^{(K)}\bigr)^{\mathsf T}
    V_{\lambda,\epsilon}^{(K)}
\end{equation}
controls the scale of the raw coordinates.
The code performs an exact symbolic $LDL^{\mathsf T}$ decomposition of this rational-function matrix, uses the induced invertible triangular congruence, and extracts the leading powers and coefficients of $K$ exactly.
The resulting renormalized block has a coefficientwise finite limit.
Every finite-$K$ congruence preserves positive semidefiniteness; consequently, every coefficientwise limit of renormalized feasible blocks is positive semidefinite whenever the interpolated rational models have been certified.

We now give explicitly the limiting SDPs at levels one and two.
The orbit variables $g_i$ and their representative words are defined in the accompanying code.
At level $t$, the limiting winning-probability bound is
\begin{equation}\nonumber
    \omega_t^{(\infty)} \,=\, \frac14
    \left(1+\max \set[\Big]{ 2g_1+g_2 \,:\, \overline B_{\lambda,\epsilon}^{(t)}(g) \succcurlyeq\zero \text{ for every occurring }(\lambda,\epsilon) } \right).
\end{equation}

\paragraph{Level one.}
The four limiting blocks are
\begin{align}
    \overline B_{(K),+}^{(1)}(g)
    &\,=\,
    \begin{pmatrix}
        1 & g_1\\
        g_1 & \frac12g_3
    \end{pmatrix}\,,
    &
    \overline B_{(K),-}^{(1)}(g)
    &\,=\,
    \begin{pmatrix}
        \frac12g_3
    \end{pmatrix}\,,
    \nonumber\\
    \overline B_{(K-1,1),+}^{(1)}(g)
    &\,=\,
    \begin{pmatrix}
        \frac12(1+g_2-g_3)
    \end{pmatrix}\,,
    &
    \overline B_{(K-1,1),-}^{(1)}(g)
    &\,=\,
    \begin{pmatrix}
        \frac12(1-g_2-g_3)
    \end{pmatrix}\,.
    \nonumber
\end{align}

\paragraph{Level two.}
The two blocks associated with the trivial representation are
$$
\resizebox{\linewidth}{!}{$
    \overline B_{(K),+}^{(2)}(g)
    \,=\,
    \begin{pmatrix}
        1 & g_1 & g_2 & g_3\\
        g_1 & \frac12g_3 & g_4 & \frac12(g_5+g_{11})\\
        g_2 & g_4 & -g_6 & g_{13}\\
        g_3 & \frac12(g_5+g_{11}) & g_{13} & \frac12(g_9+g_{18})
    \end{pmatrix}
    \quad\text{and}\quad
    \overline B_{(K),-}^{(2)}(g)
    \,=\,
    \begin{pmatrix}
        \frac12g_3 & \frac12(g_{11}-g_5) & -g_5\\
        \frac12(g_{11}-g_5) & \frac12(g_{18}-g_9) & 0\\
        -g_5 & 0 & g_9
    \end{pmatrix}.
$}
$$
The two blocks associated with the standard representation are
$$
\resizebox{\linewidth}{!}{$
    \overline B_{(K-1,1),+}^{(2)}(g)
    \,=\,
    \begin{pmatrix}
        \frac12(1+g_2-g_3)
        & g_1-g_4
        & \frac12(g_1+g_{10}-2g_{11}-2g_5)
        & \frac14(g_1-g_{10}+2g_4)
        & \frac12(g_1+g_4)
        \\
        g_1-g_4
        & 1+g_6
        & -2g_{13}
        & 0
        & g_3
        \\
        \frac12(g_1+g_{10}-2g_{11}-2g_5)
        & -2g_{13}
        & \frac12(2g_{16}+g_{17}-4g_{18}+g_3-4g_9)
        & \frac14(-g_{17}+g_3+2g_8)
        & -\frac12(g_{13}+g_{14})
        \\
        \frac14(g_1-g_{10}+2g_4)
        & 0
        & \frac14(-g_{17}+g_3+2g_8)
        & \frac18(-2g_{16}+g_{17}+g_3)
        & \frac14(g_{13}-g_{14})
        \\
        \frac12(g_1+g_4)
        & g_3
        & -\frac12(g_{13}+g_{14})
        & \frac14(g_{13}-g_{14})
        & \frac12(g_3+g_8)
    \end{pmatrix}
$}
$$
and
$$
\resizebox{\linewidth}{!}{$
    \overline B_{(K-1,1),-}^{(2)}(g)
    \,=\,
    \begin{pmatrix}
        \frac12(1-g_2-g_3)
        & \frac12(g_1+g_{10}-2g_{11}+2g_5)
        & \frac14(g_1-g_{10}-2g_4)
        & -g_1+g_4+2g_5
        \\
        \frac12(g_1+g_{10}-2g_{11}+2g_5)
        & \frac12(2g_{16}+g_{17}-4g_{18}+g_3+4g_9)
        & \frac14(-g_{17}+g_3-2g_8)
        & -g_{13}-g_{14}
        \\
        \frac14(g_1-g_{10}-2g_4)
        & \frac14(-g_{17}+g_3-2g_8)
        & \frac18(-2g_{16}+g_{17}+g_3)
        & \frac12(g_{13}-g_{14})
        \\
        -g_1+g_4+2g_5
        & -g_{13}-g_{14}
        & \frac12(g_{13}-g_{14})
        & 2g_3-2g_8-4g_9
    \end{pmatrix}.
$}
$$
The remaining four blocks are
\begin{align*}
    \overline B_{(K-2,2),+}^{(2)}(g)
    \,=\,&
    \begin{pmatrix}
        1+g_{15}-2g_{16}-g_{17}+2g_{18}-g_3+g_6+g_7+2g_9
    \end{pmatrix}\,,
    \\[1ex]
    \overline B_{(K-2,2),-}^{(2)}(g)
    \,=\,&
    \begin{pmatrix}
        1+g_{15}-2g_{16}-g_{17}+2g_{18}-g_3-g_6-g_7-2g_9
        &
        2(-g_{12}+g_{13}+g_{14}-g_2)
        \\
        2(-g_{12}+g_{13}+g_{14}-g_2)
        &
        2(1-2g_3-g_7+2g_8+2g_9)
    \end{pmatrix}\,,
    \\[1ex]
    \overline B_{(K-2,1,1),+}^{(2)}(g)
    \,=\,&
    \begin{pmatrix}
        \frac14(1-g_{15}+2g_{16}-g_{17}-g_3-g_6+g_7)
        &
        \frac12(g_2-g_{12}-g_{13}+g_{14})
        \\
        \frac12(g_2-g_{12}-g_{13}+g_{14})
        &
        \frac12(1-2g_3+g_7-2g_8)
    \end{pmatrix}\,,
    \\[1ex]
    \overline B_{(K-2,1,1),-}^{(2)}(g)
    \,=\,&
    \begin{pmatrix}
        \frac14(1-g_{15}+2g_{16}-g_{17}-g_3+g_6-g_7)
    \end{pmatrix}\,.
\end{align*}

\paragraph{Level three.}
At level three, the same construction gives fourteen limiting blocks in $175$ variables.
Their explicit entries are substantially longer and are therefore left to the accompanying code.

\end{document}